\documentclass[letterpaper, 10pt, conference]{IEEEtran}
\IEEEoverridecommandlockouts                              % This command is only needed if
\usepackage[T1]{fontenc}
\usepackage{newtxtext}
\usepackage{dsfont}
\usepackage{amsmath,amssymb,amsthm, amsfonts,mathtools,bm,eufrak}
\usepackage{algorithm}
\usepackage{algorithmicx}
\usepackage{algpseudocode}
\makeatletter
\let\IEEEsaved@labelindent\labelindent
\let\labelindent\@undefined
\makeatother
\usepackage{enumitem}
\makeatletter
\ifx\IEEEsaved@labelindent\@undefined\else
  \let\labelindent\IEEEsaved@labelindent
\fi
\makeatother
\usepackage{graphicx}
\usepackage{tikz}
\usepackage{subfig}
\usepackage[font=small]{caption}
\usepackage[hidelinks]{hyperref}
\usepackage{cite}
\usepackage{cleveref}
\usepackage{xcolor}
\usepackage{titlesec}
\usepackage{mathrsfs}
\newtheorem{assumption}{Assumption}
\newtheorem{theorem}{Theorem}
\newtheorem{proposition}{Proposition}
\newtheorem{corollary}{Corollary}
\newtheorem{lemma}{Lemma}
\newtheorem{remark}{Remark}

\usepackage[framemethod=TikZ]{mdframed}

\mdfdefinestyle{shadecommon}{
  userdefinedwidth=\linewidth,
  linewidth=0.1pt,
  linecolor=gray!45,
  roundcorner=8pt,        % rounds the border
  leftmargin=0pt,
  rightmargin=0pt,
  innerleftmargin=1.5pt,
  innerrightmargin=1.5pt,
  innertopmargin=0pt,
  innerbottommargin=2pt,
  skipabove=3pt,
  skipbelow=3pt
}

\newmdenv[
  style=shadecommon,
  backgroundcolor=gray!3
]{shadeboxlight}

\newmdenv[
  style=shadecommon,
  backgroundcolor=gray!7
]{shadeboxmedium}

\newmdenv[
  style=shadecommon,
  backgroundcolor=gray!12
]{shadeboxdark}
\newcommand{\R}{\mathbb{R}}

\newcommand{\N}{\mathbb{N}}

\newcommand{\Cc}{\mathcal{C}}

\newcommand{\ceq}{\coloneq}

\DeclareMathOperator*{\argmin}{arg\,min}

\newcommand{\dx}{\mathrm{d}}

\def\bs{\boldsymbol}
\def\mf{\mathbf}
\def\mbb{\mathbb}
\def\mc{\mathcal}
\def\ds{\mathds}
\def\mrm{\mathrm}

\def\mfk{\mathfrak}

\def\what{\widehat}
\def\mcr{\mathscr}

\newcommand{\paren}[1]{\left( #1 \right)}

\newcommand{\bmat}[1]{\begin{bmatrix}#1\end{bmatrix}}

\newcommand{\eqn}[1]{\begin{equation*}#1\end{equation*}}
\newcommand{\eqnn}[1]{\begin{equation}#1\end{equation}}
\makeatletter
\let\OldStatex\Statex
\renewcommand{\Statex}[1][3]{%
  \setlength\@tempdima{\algorithmicindent}%
  \OldStatex\hskip\dimexpr#1\@tempdima\relax}
\makeatother

\allowdisplaybreaks
\makeatletter
\def\ps@IEEEpreprintnotice{%
  \def\@oddhead{\hfil\footnotesize\preprintnotice\hfil}%
  \let\@evenhead\@oddhead
  \def\@oddfoot{}\let\@evenfoot\@oddfoot}
\makeatother
\newcommand{\preprintnotice}{This is a preprint version accepted for
  publication in the 2026 65th IEEE Conference on Decision and Control (CDC).}

\title{\Large \bf On Data-Driven Koopman Representations of Nonlinear Delay Differential Equations}

\author{ Santosh Mohan Rajkumar, Dibyasri Barman, Kumar Vikram Singh, and Debdipta Goswami
\thanks{This work is supported by NSF-DMS Math-DT under Grant 2529302.}
\thanks{ Santosh Mohan Rajkumar and Debdipta Goswami are with the Department of Mechanical and Aerospace Engineering, The Ohio State University, Columbus,
OH, USA [email:  {\{\texttt{rajkumar.36,goswami.78}\}@osu.edu}].
}
\thanks{
Dibyasri Barman and Kumar Vikram Singh are with the Department of Chemical, Paper, and Biomedical Engineering and the Department of Mechanical and Manufacturing Engineering, respectively, at Miami University, Oxford, OH, USA. [email: \{\texttt{barmand, singhkv}\}@miamioh.edu].
}
\thanks{Open-source Code: \href{https://github.com/santoshrajkumar/koopman-dde-kEDMD}{ \tiny \texttt{https://github.com/santoshrajkumar/koopman-dde-kEDMD}}
}
}

\begin{document}

\maketitle
\thispagestyle{IEEEpreprintnotice}

%%%%%%%%%%%%%%%%%%%%%%%%%%%%%%%%%%%%%%%%%%%%%%%%%%%%%%%%%%%%%%%%%%%%%%%%%%%%%%%%
\begin{abstract}

% This work establishes a rigorous bridge between infinite-dimensional delay dynamics and finite-dimensional Koopman learning, with explicit and interpretable error guarantees. While Koopman analysis is well developed for ordinary differential equations (ODEs) and partially for partial differential equations (PDEs), its extension to delay differential equations (DDEs) remains limited due to their infinite-dimensional phase space. We propose a finite-dimensional Koopman approximation framework based on history discretization and a suitable reconstruction operator, enabling a tractable representation of the Koopman operator via kernel-based extended dynamic mode decomposition (kEDMD). Deterministic error bounds are derived for the learned predictor, decomposing the total error into history discretization, kernel interpolation, and data-driven regression components. Furthermore, we develop a kernel-based reconstruction method to recover discretized states from lifted Koopman coordinates, with provable guarantees. Numerical results demonstrate convergence of the learned predictor with respect to both discretization resolution and data size, validating the proposed approach for reliable prediction for delay systems.

This work establishes a rigorous bridge between infinite-dimensional delay dynamics and finite-dimensional Koopman learning, with explicit and interpretable error guarantees. While Koopman analysis is well-developed for ordinary differential equations (ODEs) and partially for partial differential equations (PDEs), its extension to delay differential equations (DDEs) remains limited due to the infinite-dimensional phase space of DDEs. We propose a finite-dimensional Koopman approximation framework based on history discretization and a suitable reconstruction operator, enabling a tractable representation of the Koopman operator via kernel-based extended dynamic mode decomposition (kEDMD). Deterministic error bounds are derived for the learned predictor, decomposing the total error into contributions from history discretization, kernel interpolation, and data-driven regression. Additionally, we develop a kernel-based reconstruction method to recover discretized states from lifted Koopman coordinates, with provable guarantees. Numerical results demonstrate reliable prediction of nonlinear delay systems, with potential relevance to future control applications.

\end{abstract}

% \begin{IEEEkeywords}
% DDE, Koopman
% \end{IEEEkeywords}

%\vspace{-1.8em}
%%%%%%%%%%%%%%%%%%%%%%%%%%%%%%%%%%%%%%%%%%%%%%%%%%%%%%%%%%%%%%%%%%%%%%%%%%%%%%%%

% ---- begin sections/1_intro ----
\section{Introduction}
\noindent 
Delay differential equations (DDEs) arise naturally in engineering and scientific systems, including networked control systems, biological regulation, and mechanical systems with delayed feedback \cite{diekmann1995delay, rajkumar2020online, glass2021nonlinear, rihan2014delay}. Unlike ordinary differential equations (ODEs), DDEs evolve on an infinite-dimensional phase space consisting of history segments, which fundamentally complicates modeling, prediction, and analysis \cite{wang2023stochastic}. 

% Recent years have witnessed growing interest in data-driven modeling of nonlinear DDEs \cite{stephany2024learning,  zhang2026extended, chen2025dynamic, pecile2025data}. Adjoint-based neural DDE (NDDE) method in \cite{stephany2024learning} considers joint learning of dynamics and unknown delays, but incur high computational cost and training instability. The extended NDDE approach \cite{zhang2026extended} enhances expressivity by learning delays and initial conditions, but suffers from memory and scalability bottlenecks. Hybrid physics informed machine learning (ML) reduction method in \cite{chen2025dynamic} yields finite-dimensional models but incurs bias from basis truncation, limiting accuracy for complex nonlinear delay dynamics. Additionally, sparse identification methods, such as SINDy, applied to DDEs \cite{pecile2025data} provide interpretable models, but depend on carefully chosen libraries and struggle to capture complex delay interactions. Consequently, existing approaches face a fundamental trade-off between infinite-dimensional fidelity, computational tractability, and theoretical reliability.

Recent years have witnessed increasing interest in data-driven modeling of nonlinear DDEs \cite{stephany2024learning, zhang2026extended, chen2025dynamic, pecile2025data}. Adjoint-based neural DDE methods \cite{stephany2024learning} enable joint learning of dynamics and delays, but suffer from high computational cost and training instability due to backpropagation through delayed trajectories. Continuous-time extensions \cite{zhang2026extended} improve expressivity, yet inherit significant memory and scalability limitations. Hybrid physics–machine learning approaches \cite{chen2025dynamic} yield tractable finite-dimensional models, but introduce bias through basis truncation. Sparse identification methods such as SINDy \cite{pecile2025data} provide interpretable representations, but rely heavily on library design and struggle to capture complex delay interactions. Consequently, existing approaches face a fundamental trade-off between infinite-dimensional fidelity, computational tractability, and theoretical reliability.

In parallel, the Koopman operator framework has emerged as a powerful paradigm for data-driven analysis and control of nonlinear dynamical systems \cite{koopman1931hamiltonian,rowley2009spectral, goswami2017global, korda2018linear,goswami2021bilinearization, strasser2026overview, Nayak2024tcKAE, rajkumar2025real}. For nonlinear systems governed by ODEs, data-driven Koopman methods such as extended dynamic mode decomposition (EDMD) \cite{williams2015data} and its kernel variants (kEDMD) have been widely developed, with recent works establishing deterministic error bounds for nonlinear systems \cite{kohne2025error, bold2025kernel, strasser2025kernel}. Extensions to PDEs have also been explored \cite{nakao2020spectral, peitz2025equivariance}, but these works primarily focus on the spectral approximation of compact differential operators without error bounds suitable for reliable prediction and control design.

% Extending Koopman-based methods to DDEs remains fundamentally challenging due to the infinite-dimensional nature of the state, which consists of history functions rather than finite-dimensional vectors. Existing Koopman-based approaches to delay systems \cite{brunton2017chaos} rely on heuristic delay embeddings that compress the infinite-dimensional history into finitely many samples, without explicitly quantifying the induced approximation error. As a result, these methods lack guarantees required for reliable prediction and control.

Extending Koopman-based methods to DDEs remains challenging because the state is an infinite-dimensional history function rather than a finite-dimensional vector. Delay-coordinate Koopman formulations \cite{brunton2017chaos} use finite collections of time-delayed measurements to obtain tractable representations, but do not provide approximation error bounds associated with the finite delay-coordinate representation. For DDEs with time-varying delays, \cite{muller2017dynamical} developed an operator-theoretic framework that decomposes the solution operator into a Koopman composition operator associated with delayed access and a constant-delay evolution operator, primarily for analyzing Lyapunov spectra and related spectral properties. More recently, \cite{han2023robust} proposed a deep recurrent Koopman framework that uses an LSTM to encode sequences of past states and control inputs into probabilistic observables for prediction and robust control of nonlinear systems with unknown time delays. To the best of our knowledge, however, a quantitative bridge from infinite-dimensional DDE dynamics to finite-dimensional data-driven Koopman representations with explicit approximation error bounds has not yet been established.

This paper addresses this gap by developing a theoretically grounded, data-driven Koopman framework for nonlinear DDEs with explicit and interpretable error guarantees. The key idea is to construct a finite-dimensional realization of the infinite-dimensional delay dynamics via history discretization and a suitable reconstruction operator, enabling a principled connection between the Koopman operator on the space of history functions and a finite-dimensional data-driven representation. The main contributions are summarized as follows:
(i) \textit{Koopman realization for DDEs:} We establish a rigorous
bridge between infinite-dimensional delay dynamics and
finite-dimensional Koopman representations through a
sampling–reconstruction framework, yielding a tractable
surrogate model.
(ii) \textit{Data-driven linear surrogate:} Using kernel-based extended dynamic mode decomposition (kEDMD) with Wendland radial basis function (RBF) kernels, we construct a finite-dimensional linear surrogate model that rigorously approximates the infinite-dimensional Koopman operator defined on the Banach space of history functions through a data-driven matrix representation.
(iii) \textit{Deterministic error guarantees:} We derive explicit error bounds for the learned Koopman predictor, decomposing the total approximation error into history discretization, kernel interpolation, and data-driven regression components, providing insight into approximation accuracy.
(iv) \textit{State reconstruction with guarantees:} We develop a kernel-based mechanism to recover the finite-dimensional discretized state from lifted Koopman coordinates and establish deterministic bounds on the resulting state prediction error, supporting reliable simulation and potential control applications.

\noindent \textbf{Notations:} Let $\N:=\{1,2,\cdots\}$, $\N_0 := \{0\}\cup\N$, and $[a:b]:=[a,b] \cap \N_0$ for any $a,b \in \N_0$. Unless otherwise stated, $\|\cdot\|$ denotes the Euclidean norm for vectors in $\mathbb{R}^n$ and the induced matrix $2$-norm (spectral norm) for matrices in $\mathbb{R}^{n \times n}$, with dimensions clear from context.  For a maximal finite time-delay $\tau_d > 0 $, we denote by $\mathscr{C} \ceq \Cc \paren{[-\tau_d,0], \R^n}$
the Banach space of continuous functions $\eta:[-\tau_d,0] \to \R^n$ with the uniform norm $\| \eta \|_\infty := \sup_{s \in[-\tau_d,0]}\| \eta(s)\|$. For any $z = [z_1^\top \ \cdots \ z_M^\top]^\top \in \mathbb{R}^{nM}$ with $z_i \in \mathbb{R}^n$ for $i \in [1:M]$, define the block supremum norm as
$\|z\|_{b,\infty} := \max_{i \in [1:M]} \|z_i\|.$ For any matrix,  $\|\cdot\|_F$ denotes the Frobenius norm.
% ---- end sections/1_intro ----
% ---- begin sections/2_prelims ----
\section{Preliminaries}

\subsection{Delay Differential Equations (DDEs)}

Consider a nonlinear retarded autonomous delay differential equation (DDE) of the form \cite{hale1993introduction}
\eqnn{
\dot x(t) = f\paren{x_t},\quad t \geq 0,
\label{eq:dde1}
}
where $x(t) \in \R^n$ is the \textit{current state} and $x_t \in \mcr C$ is the \textit{history state} defined by $x_t(\theta) \ceq x(t+\theta)$, $\theta \in [-\tau_d,0]$. Let $\Omega \subset \mcr C$ be open and $f:\Omega \to \R^n$ be continuous and locally Lipschitz on $\Omega$. Then, for every initial history $\eta \in \Omega$, there exists a unique maximal forward solution $x(t;\eta)$, defined for $t \in [-\tau_d,t_\alpha{\scriptstyle(\eta)})$, with $t_\alpha{\scriptstyle(\eta)} \in (0,+\infty]$, satisfying $x(\theta;\eta)=\eta(\theta)$ for $\theta \in[-\tau_d,0]$. Consequently, $x_t(\theta;\eta)=x(t+\theta;\eta)$ is defined for $t \in [0,t_\alpha{\scriptstyle(\eta)})$ and $\theta \in [-\tau_d,0]$.

% \begin{assumption}\label{assm1}
%     Let the compact operational domain be $W \! := \! \{ \eta \in \Omega \mid \|\eta\|_\infty \leq \gamma \text{ and } 
%      \| \eta(\theta_1)-\eta(\theta_2) \| \leq L_\theta |\theta_1-\theta_2| , \forall\  \theta_1,\theta_2 \! \in \! [-\tau_d,0]\}$, for some constants $\gamma, L_\theta >0$. We assume that $W$ is forward-invariant under the semi-flow $\set{\Phi^t}_{t\geq 0}$ generated by \eqref{eq:dde1}, i.e., $\Phi^t(W) \subseteq W$, for $t \in \mbb T$,  where $\mbb T := [0,t_f]$ and $0<t_f < t_\alpha$ (c.f. \cite[Section 2.2]{hale2006functional}).
%      % Moreover, there exists a finite constant $L_F > 0$ such that $\|\Phi^t(\eta_1) - \Phi^t(\eta_2)\|_\infty \leq L_F \|\eta_1 - \eta_2\|_\infty$ for all $\eta_1, \eta_2 \in W$ and $t \in \mathbb{T}$.
% \end{assumption}

\begin{shadeboxlight}
\begin{assumption}\label{assm1}
    Let the compact operational domain be $W \! := \! \{ \eta \in \mcr C \mid \|\eta\|_\infty \leq \gamma \text{ and } 
     \| \eta(\theta_1)-\eta(\theta_2) \| \leq L_\theta |\theta_1-\theta_2| , \forall\  \theta_1,\theta_2 \! \in \! [-\tau_d,0]\}$, for some constants $\gamma, L_\theta >0$, and assume that $W \subset \Omega$.  Define
$t_W\ceq\inf_{\eta\in W}t_\alpha(\eta)$. We assume that there exists
$t_f\in(0,t_W)$ such that $W$ is forward invariant under the local
semiflow $\{\Phi^t\}_{t\geq0}$ generated by \eqref{eq:dde1} on
$\mbb T\ceq[0,t_f]$, i.e., $\Phi^t(W)\subseteq W$ for all
$t\in\mbb T$ (cf. \cite[Section 2.2]{hale1993introduction}).
\end{assumption}
\end{shadeboxlight}

\subsection{Finite-Dimensional Realization}

To obtain a finite-dimensional representation of any state function, 
let $\Theta_M := \{\theta_j\}_{j=1}^M \subset [-\tau_d,0]$, with $M\ge2$, denote a set of 
uniformly spaced discretization points with $\theta_1=-\tau_d$, $\theta_M=0$, and grid 
spacing $\delta(M) := \tau_d/(M-1)$. Define the finite-dimensional realization of $W$ as the compact operational domain $\mc Z \subset \R^{nM}$ defined as $\mc Z := \{ \zeta \in \R^{nM} \mid \| \zeta \|_{b,\infty} \le \gamma \text{ and } \|\zeta_{j+1}-\zeta_j \| \le L_\theta \delta(M),\ j=1,\cdots,(M-1) \}$, where $\zeta = [\zeta_1^\top \ \cdots \ \zeta_M^\top]^\top$ with $\zeta_j \in \mathbb{R}^n$. Define the sampling operator $\mathcal{Q}: W \to \mc Z$ as
\begin{equation}
\mathcal{Q}(\eta) :=
[\eta(\theta_1)^\top \ \eta(\theta_2)^\top \ \cdots \ \eta(\theta_M)^\top]^\top, \ \forall \ \eta \in W.
\end{equation}
Then $\mathcal{Q}(W) \subseteq \mathcal{Z}$, since every $\eta \in W$ satisfies the defining constraints of $\mathcal Z$. 
% For each $k \in \mbb T_k$, the discretized state vector is given by $z_k := \mathcal{Q}(\phi_k) \in \mathcal{Z}$.
Now, define the reconstruction operator $\mathcal{R}: \mathcal{Z} \to \mcr C$ for any $\zeta \in \mathcal{Z}$, $\mc R$ is defined as
\begin{equation}
\mathcal{R}(\zeta)(\theta) := \zeta_j + \tfrac{\theta - \theta_j}{\delta(M)} (\zeta_{j+1} - \zeta_j), \quad \forall \theta \in [\theta_j, \theta_{j+1}],
\label{eq:recon}
\end{equation}
and $j \in [1:M-1]$. Although $\mathcal{R}(\zeta)(\cdot)$ is piecewise affine in the argument $\theta$, the operator $\mathcal{R}$ is linear with respect to $\zeta$. Specifically, $\mathcal{R}$ is the restriction of a bounded linear operator defined on $\mathbb{R}^{nM}$. Therefore, $\mathcal{R}$ is Fr\'echet differentiable with derivative $D\mathcal{R}(\zeta) = \mathcal{R}$ (independent of $\zeta$) and second derivative $D^2\mathcal{R} = 0$.

\begin{shadeboxmedium}
\begin{lemma}\label{lem1} The sampling operator $\mc Q$ and the reconstruction operator $\mc R$ satisfy:
\begin{enumerate}
    \item [\textup{(i)}] $\mc R$ maps $\mathcal{Z}$ into $W$, i.e., $\mathcal{R}(\mathcal{Z}) \subseteq W$. Moreover,
    $\|\mc R(\zeta)-\mc R(\xi)\|_\infty \le \|\zeta-\xi\|_{b,\infty},
    \ \forall \zeta,\xi\in\mc Z$.  
    \item [\textup{(ii)}] The composition $\mc Q \circ \mc R$ is the identity on $\mc Z$.
    \item [\textup{(iii)}] The composition $\mc P := \mc R \circ \mc Q$ is a projection on $W$ satisfying
    $
    \|\eta - \mc P(\eta)\|_\infty \leq L_\theta \delta(M), \quad \forall \eta \in W.
    $
\end{enumerate}
\end{lemma}
\end{shadeboxmedium}

\begin{proof}
   See Appendix (\Cref{sec:proofs}).  
\end{proof}

\subsection{Problem Setting}

Let $\Delta \in (0,t_f]$ denote a fixed sampling interval, and assume the maximal time delay $\tau_d$ is known. The corresponding discrete-time semi-flow is given by the time-$\Delta$ map $F := \Phi^\Delta : W \to W$. Let $\mbb T_k := \{k \in \N_0 \mid (k+1)\Delta \le t_f\}$. For each \(k\in\mbb T_k\), let
\(\phi_k:=x_{k\Delta}\) and \(\phi_{k+1}:=x_{(k+1)\Delta}\). Therefore, for each $k \in \mbb T_k$, the sampled history state evolves according to
\eqnn{
\phi_{k+1} = F(\phi_k).
\label{eq:dde_samp}
}
Define the discretized history state as $z_k := \mc Q(\phi_k)$. Then, the exact evolution of $z_k$ is given by
\eqnn{
z_{k+1} = \mc Q(\phi_{k+1}) = \mc Q \circ F(\phi_k).
\label{eq:dde_disc}
}
Equation \eqref{eq:dde_disc} is exact but is not generally closed on $\mc Z$, because $z_k$ contains only finitely many samples of the history and therefore does not uniquely determine $\phi_k$. To obtain a transition map on $\mc Z$, we replace $\phi_k$
by its reconstruction $\mc R(z_k)$ and define $\tilde F := \mc Q \circ F \circ \mc R$.
Then, with $\mc P=\mc R\circ \mc Q$ and $z_k=\mc Q(\phi_k)$, we have
$\tilde F(z_k)= \mc Q \circ F \circ \mc R(\mc Q(\phi_k)) = \mc Q \circ F \circ \mc P (\phi_k)$. Since generally $\mc P\neq I$ on $W$,
this yields only an approximation of the exact discretized evolution
\eqnn{
z_{k+1} \approx \tilde F(z_k).
\label{eq:dde_ind}
}
The underlying map $F$ is unknown, and consequently, the induced map $\tilde F$ is also unknown.
We assume access only to snapshot pairs $(z_k, z_{k+1})$ corresponding to the true history state evolution pairs $(\phi_k, \phi_{k+1})$. 

%with $z  = \mc Q(\phi)$ the map becomes $\tilde F(\cdot) = \mc Q \circ F \circ \mc P(\cdot)$.

\begin{shadeboxlight}
\begin{assumption}\label{assm2}
Assume that $f$ is twice continuously Fr\'echet differentiable on an open neighborhood of $W$. Then, by smooth dependence on the initial history, the time-$\Delta$ map $F:=\Phi^\Delta$ is twice continuously Fr\'echet differentiable on a possibly smaller open neighborhood of $W$. Consequently, there exist constants $L_F, \bar D_F > 0$ such that for all $\eta,\eta_1,\eta_2 \in W$,
$\|F(\eta_1)-F(\eta_2)\|_\infty \le L_F \|\eta_1-\eta_2\|_\infty$,
and
$\|D^2F(\eta)\|_{\mathcal L^2(\mathscr C,\mathscr C)} \le \bar D_F$,
where $\mathcal L^2(\mathscr C,\mathscr C)$ denotes the space of bounded bilinear operators on $\mathscr C$, endowed with the operator norm induced by $\|\cdot\|_\infty$.
\end{assumption}
\end{shadeboxlight}

\begin{shadeboxmedium}
\begin{lemma}\label{lem:tildeF_hessian}
Under Assumptions \ref{assm1} and \ref{assm2}, the map $\widetilde F:=\mc Q\circ F\circ\mc R:\mc Z\to\mc Z$ is the restriction of a twice continuously Fr\'echet differentiable map defined on an open neighborhood of $\mc Z$. Moreover, $\|D^2\widetilde F(z)\|_{\mathcal L^2(\mathbb R^{nM},\mathbb R^{nM})}\le D_F:=\sqrt{M}\,\bar D_F$ for all $z\in\mc Z$, where $\mathcal L^2(\mathbb R^{nM},\mathbb R^{nM})$ denotes the space of bounded bilinear maps from $\mathbb R^{nM}\times\mathbb R^{nM}$ to $\mathbb R^{nM}$, endowed with the operator norm induced by the Euclidean norm.
\end{lemma}
\end{shadeboxmedium}

\textit{Proof:} See Appendix.

% \begin{proof}
%     See Appendix (\Cref{proof:lemma1}).
% \end{proof}

We define the dataset $\mc T$ as the collection of $N = N_{tr} N_t$ transition pairs of discretized states:
\eqnn{
\mc T := \left\{ (z^{(i)}_k, z^{(i)}_{k+1}) \;\middle|\; i \in [1:N_{tr}], \; k \in [0:N_t-1] \right\},
}
where $N_{tr}$ is the number of independent trajectories and
$N_t\Delta\le t_f$. In the following, we use the shorthand $(\phi_{k+1}, z_{k+1}, \phi_k, z_k) \equiv (\phi^+, z^+, \phi, z)$. 
Given the transition dataset $\mc T$, our goal is to construct a data-driven linear surrogate model on the discretized state space $\mc Z$ that provides a finite-dimensional representation of the sampled DDE dynamics \eqref{eq:dde_samp}. This is achieved through the induced map $\tilde F$ in \eqref{eq:dde_ind}, which yields a closed approximate evolution on $\mc Z$. We further seek to establish rigorous error bounds for the learned surrogate relative to the exact discretized evolution \eqref{eq:dde_disc}.
% Given the transition dataset $\mc T$, our goal is to construct a data-driven linear surrogate model for the induced approximate evolution \eqref{eq:dde_ind} on $\mc Z$. This yields a finite-dimensional representation of the sampled DDE dynamics \eqref{eq:dde_samp} through the discretized state $z_k$. In addition, we aim to derive rigorous error bounds for the learned surrogate with respect to the exact discretized evolution \eqref{eq:dde_disc}.
% It is convenient to arrange $\mc T$ into snapshot matrices $\mathbf{Z}, \mathbf{Z}^+ \in \R^{nM \times N}$. We define these matrices by concatenating the trajectory data column-wise:
% \eqnn{
% \mathbf{Z} := \begin{bmatrix} 
% z^{(1)}_0 & \cdots & z^{(1)}_{N_t-1} & \cdots & z^{(N_{tr})}_0 & \cdots & z^{(N_{tr})}_{N_t-1} 
% \end{bmatrix},
% }
% \eqnn{
% \mathbf{Z}^+ := \begin{bmatrix} 
% z^{(1)}_1 & \cdots & z^{(1)}_{N_t} & \cdots & z^{(N_{tr})}_1 & \cdots & z^{(N_{tr})}_{N_t} 
% \end{bmatrix}.
% }
% The $j$-th column of $\mathbf{Z}$ and $\mathbf{Z}^+$ corresponds to a pair $(z_j, z_j^+)$ in the dataset, where $j \in [1: N]$. Since each $z^{(i)}_k$ is sampled from a valid history $\phi^{(i)}_k \in W$, all columns of $\mathbf{Z}$ and $\mathbf{Z}^+$ lie within the constrained set $\mc Z$.
Further, we consider a set of $p \in \N$ pairwise distinct points 
$\mbb X = \{ z_\ell \}_{\ell=1}^{p} \subset \mc Z$. Define the \textit{fill distance} of $\ds X$ over $\mc Z$ as
$h_{\mbb X} := \sup_{z \in \mc Z} \min_{z_\ell \in \ds X} \|z - z_\ell\|$. Moreover,  for each $z_\ell \in \mbb X$, we construct a local collection of $d \ge nM+1$ transition pairs 
$\mbb X_\ell = \{ (z_{\ell,j}, z_{\ell,j}^+) \}_{j=1}^d \subset \mc T$. 
Specifically, the points $\{ z_{\ell,j} \}_{j=1}^d$ are chosen as the $d$ nearest neighbors within a radius $\rho$
of $z_\ell$, and the corresponding successors 
$z_{\ell,j}^+$ are taken from the associated transition pairs in $\mc T$. For each $\ds X_\ell$ we arrange the transition pairs in matrix form as $(\bar{\mf Z}_\ell, \mf Z^+_\ell)$ with $\mf Z_\ell^+ = [z^+_{\ell,1}\ \cdots \ z^+_{\ell,d}] \in \R^{nM\times d}$, and 
\begin{align}
    \bar{\mf Z}_\ell &= \bmat{  1 & \cdots & 1 \\ z_{\ell,1}-z_\ell & \cdots & z_{\ell,d}-z_\ell } \in \R^{(nM+1)\times d}.
\end{align}

\begin{shadeboxlight}
% \begin{assumption}
% \label{assm3}
% There exists a constant $h_0 > 0$ such that the fill distance $h_{\ds X}$ satisfies $h_{\ds X} < h_0$. Moreover, for each $\ell \in [1:p]$ and the chosen $\rho$, $\mrm{rank}(\bar{\mf Z}_\ell) = nM+1$.
% \end{assumption}
\begin{assumption}
\label{assm3}
For each $\ell \in [1:p]$, the local data matrix
$\bar{\mf Z}_\ell$ has full row rank, i.e.,
$\mrm{rank}(\bar{\mf Z}_\ell)=nM+1$.
\end{assumption}
\end{shadeboxlight}

\subsection{Reproducing Kernel Hilbert Space (RKHS)}

Consider $\mrm k:\R^{nM} \times \R^{nM} \to \R$ to be a continuous strictly positive-definite symmetric kernel and define the canonical feature $\varphi_z(\cdot) = \mrm k(z,\cdot)$ for $z \in \R^{nM}$. For the distinct points $\ds X$, the kernel matrix is defined as
\begin{align}
K_{\ds X} =
\begin{bmatrix}
\mrm k(z_1,z_1) & \cdots & \mrm k(z_p,z_1)\\
\vdots & \ddots & \vdots \\
\mrm k(z_1,z_p) & \cdots & \mrm k(z_p,z_p)
\end{bmatrix}.
\end{align} 
The kernel $\mrm k$ defines a Hilbert space of functions $\mc H$ satisfying the reproducing property $g(z) = \langle g, \varphi_z \rangle_{\mc H}$ for all $g \in \mc H$ and $z \in \R^{nM}$ \cite{bold2025kernel}. The native RKHS norm $\| \cdot \|_{\mc H}$ is defined via the associated abstract inner product and is typically difficult to evaluate explicitly \cite{strasser2025kernel, wendland2004scattered}.  We consider the RKHS $\mc H$ generated by Wendland radial basis function (RBF) $\mc W_{s}:\R^{nM} \to \R$ with smoothness degree $s \in \N_0$ defined as
$\mc W_s(\xi) := \Phi_{nM,s}(\|\xi\|)$ for any $\xi \in \R^{nM}$
with $\Phi_{nM,s} \in \mc C^{2s}([0,\infty),\R)$ given as  \cite{wendland2004scattered}
\eqnn{
 \Phi_{nM,s}(r) = \begin{cases}
    P_{\mfk d, s}(r) & 0 \leq r \le 1 \\
    0, & r>1
\end{cases},
}
where $P_{\mfk d,s}$ is a univariate polynomial of degree $\mfk d = \lfloor nM/2 \rfloor+3s+1$. The induced kernel with Wendland RBF is given by $\mrm k (\zeta_1,\zeta_2) = \Phi_{nM,s}(\| \zeta_1-\zeta_2 \|)$ for $\zeta_1, \zeta_2 \in \R^{nM}$. In this work, we consider the Wendland RBF $\Phi_{nM,s} \in \mc C^{2}([0,\infty),\R)$ with $P_{\mfk d, 1}(r) = (1-r)^{\mfk d-1}[(\mfk d -1)r+1]$ with $\mfk d = \lfloor nM/2 \rfloor+4 $ and $s=1$ \cite[Corollary 9.14]{wendland2004scattered}. For Wendland kernels, the equivalence between the native RKHS norm and the fractional Sobolev norm is established via a bounded extension operator \cite{wendland2004scattered}.

% ---- end sections/2_prelims ----
% ---- begin sections/3_koopman ----
\section{Koopman Formalism and kEDMD}

Consider the space of Lipschitz continuous functionals on $W$ be defined as \cite{berninger2003lipschitz}
\eqnn{
    \mathrm{Lip}(W) :=
    \left\{
    \psi : W \to \mathbb{R}
    \;\mid\;
    \|\psi\|_{\mathrm{Lip}_W} < \infty
    \right\},
    \label{eq:norm_lipw}
}
where $\|\psi\|_{\mathrm{Lip}_W} := \|\psi\|_\infty + \mrm{Lip}_W(\psi)$, with
\eqn{
\|\psi\|_\infty := \sup_{\eta \in W} |\psi(\eta)|,\ \mrm{Lip}_W(\psi):= \sup_{\eta_1 \neq \eta_2}
    \tfrac{|\psi(\eta_1) - \psi(\eta_2)|}{\|\eta_1 - \eta_2\|_\infty}.
}
Further, let the space of Lipschitz continuous functions on $\mc Z$ defined as
\eqnn{
\mrm{Lip}(\mc Z)
:=
\left\{
\tilde{\psi}:\mc Z \to \R
\;\mid\;
\|\tilde{\psi}\|_{\mrm{Lip}(\mc Z)} < \infty
\right\}.
\label{eq:norm_lipz}
}
where $
\|\tilde{\psi}\|_{\mrm{Lip}_\mc{Z}}
:=
\|\tilde{\psi}\|_{\infty}
+
\mrm{Lip}_{\mc Z}(\tilde{\psi}),
$
with
\eqn{
\|\tilde{\psi}\|_\infty := \sup_{\zeta \in \mathcal{Z}} |\tilde{\psi}(\zeta)|,\ \mrm{Lip}_{\mc Z}(\tilde{\psi})
:=
\sup_{\zeta_1 \neq \zeta_2}
\tfrac{|\tilde{\psi}(\zeta_1)-\tilde{\psi}(\zeta_2)|}{\|\zeta_1-\zeta_2\|_{b,\infty}}.
}
From the definitions in \eqref{eq:norm_lipw} and \eqref{eq:norm_lipz}, $\mrm{Lip}(W)$ and $\mrm{Lip}(\mc Z)$ are Banach spaces, which we use as the observable spaces associated with the infinite-dimensional DDE dynamics and the finite-dimensional realization, respectively.

\begin{shadeboxmedium}
\begin{lemma}\label{lem:induced_obs}
Let $\psi \in \mrm{Lip}(W)$ and define the induced observable function
$\tilde{\psi}: \mc Z \to \R$ as $\tilde{\psi} := \psi \circ \mc R$.
Then, $\tilde{\psi} \in \mrm{Lip}(\mc Z)$ and
$
\|\tilde{\psi}\|_{\mrm{Lip}_{\mc Z}}
\le
\|\psi\|_{\mrm{Lip}_W}.
$
\end{lemma}
\end{shadeboxmedium}

\begin{proof}
Since $\mc R(\mc Z)\subseteq W$ by Lemma~1(i), the composition
$
\tilde{\psi} = \psi \circ \mc R
$
is well defined on $\mc Z$.
Now,
\begin{align*}
\|\tilde{\psi}\|_\infty
=
\sup_{\zeta\in\mc Z} |\psi(\mc R(\zeta))|
& =
\sup_{\eta\in \mc R(\mc Z)} |\psi(\eta)|\\
&\le
\sup_{\eta\in W} |\psi(\eta)|
=
\|\psi\|_\infty.
\end{align*}
Next, let $\zeta,\xi \in \mc Z$ with $\zeta \neq \xi$. Then, using \Cref{lem1}(i)
\begin{align*}
|\tilde{\psi}(\zeta)-\tilde{\psi}(\xi)|
&=
|\psi(\mc R(\zeta))-\psi(\mc R(\xi))|\\
&\le
\mrm{Lip}_W(\psi)\,\|\mc R(\zeta)-\mc R(\xi)\|_\infty\\
&\le
\mrm{Lip}_W(\psi)\,\|\zeta-\xi\|_{b,\infty}.
\end{align*}
Dividing by $\|\zeta-\xi\|_{b,\infty}$ and taking the supremum over all $\zeta \neq \xi$ gives
$
\mrm{Lip}_{\mc Z}(\tilde{\psi}) \le \mrm{Lip}_W(\psi).
$
Therefore, $\tilde{\psi} \in \mrm{Lip}(\mc Z)$ and
$ \|\tilde{\psi}\|_{\mrm{Lip}_{\mc Z}}
=
\|\tilde{\psi}\|_\infty + \mrm{Lip}_{\mc Z}(\tilde{\psi})
 \le
\|\psi\|_\infty + \mrm{Lip}_W(\psi)  = \|\psi\|_{\mrm{Lip}_W}$.
\end{proof}

\begin{remark}
Let $\mc H_{\mc Z}:=\{g|_{\mc Z}\mid g\in\mc H\}$. In the subsequent analysis,
we write $\mc H$ for $\mc H_{\mc Z}$ whenever the domain is $\mc Z$.
\end{remark}

% For any observable functional $\psi \in \mrm{Lip}(W)$, consider the corresponding induced observable $\Psi \in \mrm{Lip}(\mc Z)$ given as
% \eqnn{
% \Psi := \psi \circ \mc R,
% }
% that is,
% $\Psi(\zeta)=\psi(\mc R(\zeta)), \ \forall \zeta\in\mc Z
% $. Thus, $\Psi$ evaluates the original functional $\psi$ on the reconstructed state corresponding to the discretized state $\zeta$. Consequently, for any $\eta \in W$,
% \eqnn{
% \Psi(\mc Q(\eta))=\psi\left(\mc R\big(\mc Q(\eta)\right)\big)=\psi(\mc P(\eta)).
% }
% Hence, if $\psi \in \mrm{Lip}(W)$, then by Lemma 1(iii),
% \begin{align}
% |\psi(\eta)-\Psi(\mc Q(\eta))|
% & =
% |\psi(\eta)-\psi(\mc P(\eta))|\\
% &\le
% \mrm{Lip}_W(\psi)\|\eta-\mc P(\eta)\|_\infty\\
% &\le
% \mrm{Lip}_W(\psi)L_\theta \delta(M),
% \ \forall \eta\in W.
% \end{align}

\begin{shadeboxmedium}
\begin{lemma}\label{lem:rkhs_lip}
The RKHS $\mc H$ is continuously embedded in $\mrm{Lip}(\mc Z)$. In particular,
there exists a constant $C_{\mc H}>0$ such that
$
\|g\|_{\mrm{Lip}(\mc Z)}
\le C_{\mc H}\|g\|_{\mc H},
\ \forall g\in\mc H.
$
\end{lemma}
\end{shadeboxmedium}

\begin{proof}
Let $g\in\mc H$. Since $\mc H$
is an RKHS with reproducing kernel $\mrm k$, for any $\zeta\in\mc Z$,
$
g(\zeta)
=
\langle g,k(\zeta,\cdot)\rangle_{\mc H}.
$
Hence,
$
|g(\zeta)|
\le
\|g\|_{\mc H}\|\mrm k(\zeta,\cdot)\|_{\mc H}
=
\|g\|_{\mc H}\sqrt{\mrm k(\zeta,\zeta)}.
$
Since $\mrm k$ is continuous and $\mc Z$ is compact,
$
C_0
:=
\sup_{\zeta\in\mc Z}\sqrt{k(\zeta,\zeta)}
<\infty,
$
and therefore
$
\|g\|_\infty
\le
C_0\|g\|_{\mc H}.
$
Next, for any $\zeta_1,\zeta_2\in\mc Z$,
\begin{align*}
|g(\zeta_1)-g(\zeta_2)|
& =
\left|
\langle
g,
k(\zeta_1,\cdot)-k(\zeta_2,\cdot)
\rangle_{\mc H}
\right|\\
& \le
\|g\|_{\mc H}
\|k(\zeta_1,\cdot)-k(\zeta_2,\cdot)\|_{\mc H}.
\end{align*}
Since $\mrm k$ is $\mc C^2$, the feature map
$\zeta\mapsto \mrm k(\zeta,\cdot)$ is Lipschitz on the compact set
$\mc Z$. Thus, there exists a constant $C_1>0$ such that
\[
\|k(\zeta_1,\cdot)-k(\zeta_2,\cdot)\|_{\mc H}
\le
C_1\|\zeta_1-\zeta_2\|_{b,\infty}.
\]
Consequently,
$
\mrm{Lip}_{\mc Z}(g)
\le
C_1\|g\|_{\mc H}.
$
Therefore, 
\[
\|g\|_{\mrm{Lip}(\mc Z)}
=
\|g\|_\infty+\mrm{Lip}_{\mc Z}(g)
\le
(C_0+C_1)\|g\|_{\mc H}.
\]
Setting $C_{\mc H}:=C_0+C_1$ gives the stated result.
\end{proof}

% \begin{proof}
% Let $g\in\mc H$ and denote its restriction to $\mc Z$ by
% $g|_{\mc Z}$. Since $\mc H$ is an RKHS with reproducing kernel $\mrm k$, for any $\zeta\in\mc Z$,\[ g|_{\mc Z}(\zeta)=g(\zeta) = \langle g,k(\zeta,\cdot)\rangle_{\mc H}. \] . Hence,
% \[
% |g|_{\mc Z}(\zeta)| \le \|g\|_{\mc H}\|k(\zeta,\cdot)\|_{\mc H} = \|g\|_{\mc H}\sqrt{k(\zeta,\zeta)}.
% \]
% Since $\mrm k$ is continuous and $\mc Z$ is compact, there exists
% \[
% C_0:=\sup_{\zeta\in\mc Z}\sqrt{\mrm k(\zeta,\zeta)}<\infty, \text{ so that}
% \]
% $\|\Psi\|_\infty \le C_0\|\Psi\|_{\mc H}$. Next, for any $\zeta_1,\zeta_2\in\mc Z$, the reproducing property gives
% \begin{align*}
% |\Psi(\zeta_1)-\Psi(\zeta_2)|
% &=
% |\langle \Psi,\mrm k(\zeta_1,\cdot)-\mrm k(\zeta_2,\cdot)\rangle_{\mc H}|\\
% &\le
% \|\Psi\|_{\mc H}\,
% \|\mrm k(\zeta_1,\cdot)-\mrm k(\zeta_2,\cdot)\|_{\mc H}.
% \end{align*}
% Since the kernel is $\mc C^2$, the feature map $\zeta\mapsto \mrm k(\zeta,\cdot)$ is Lipschitz on
% the compact set $\mc Z$, and thus there exists $C_1>0$ such that
% \[
% \|\mrm k(\zeta_1,\cdot)-\mrm k(\zeta_2,\cdot)\|_{\mc H}
% \le C_1\|\zeta_1-\zeta_2\|_{b,\infty},
% \qquad \forall \zeta_1,\zeta_2\in\mc Z.
% \]
% Therefore,
% $
% |\Psi(\zeta_1)-\Psi(\zeta_2)|
% \le C_1\|\Psi\|_{\mc H}\|\zeta_1-\zeta_2\|_{b,\infty},
% $
% which gives
% $
% \mrm{Lip}_{\mc Z}(\Psi)\le C_1\|\Psi\|_{\mc H}.
% $
% Combining the above estimates yields
% $
% \|\Psi\|_{\mrm{Lip}_{\mc Z}}
% =
% \|\Psi\|_\infty+\mrm{Lip}_{\mc Z}(\Psi)
% \le C_{\mc H}\|\Psi\|_{\mc H},
% $
% where $C_{\mc H}:=C_0+C_1$. Therefore, $\mc H \hookrightarrow \mrm{Lip}(\mc Z)$.
% \end{proof}

\begin{shadeboxmedium}
\begin{corollary}\label{cor1}
For any $g\in\mc H$, the induced functional
$\widehat{g}:=g\circ\mc Q:W\to\R$ belongs to
$\mrm{Lip}(W)$ and satisfies
$
\|\widehat{g}\|_{\mrm{Lip}(W)}
\le
\|g\|_{\mrm{Lip}(\mc Z)}
\le
C_{\mc H}\|g\|_{\mc H}.
$
\end{corollary}
\end{shadeboxmedium}

\begin{proof}
For any distinct $\eta_1,\eta_2\in W$, 
$|\widehat{g}(\eta_1)-\widehat{g}(\eta_2)|
=|g(\mc Q(\eta_1))-g(\mc Q(\eta_2))|
\le \mrm{Lip}_{\mc Z}(g)\|\mc Q(\eta_1)-\mc Q(\eta_2)\|_{b,\infty}
\le \mrm{Lip}_{\mc Z}(g)\|\eta_1-\eta_2\|_\infty$.
Dividing by $\|\eta_1-\eta_2\|_\infty$ and taking the supremum over $\eta_1\neq\eta_2$ yields $\mrm{Lip}_{W}(\widehat{g})\le \mrm{Lip}_{\mc Z}(g)$. Moreover, $\|\widehat{g}\|_\infty\le \|g\|_\infty$, and hence $\|\widehat{g}\|_{\mrm{Lip}(W)}\le \|g\|_{\mrm{Lip}(\mc Z)}$, so $\widehat{g}\in\mrm{Lip}(W)$. Finally, by \Cref{lem:rkhs_lip}, $\|\widehat{g}\|_{\mrm{Lip}(W)}\le \|g\|_{\mrm{Lip}(\mc Z)}\le C_{\mc H}\|g\|_{\mc H}$.
\end{proof}

\subsection{Koopman Operator}

The Koopman operator is a linear but infinite-dimensional operator that enables the study of nonlinear dynamical systems through its action on observables defined on a Banach space \cite{strasser2026overview}. The action of the Koopman operator $\mc K$ on any $\psi \in \mrm{Lip}(W)$ is given by
\begin{align}
(\mc K \psi)(\phi) = \psi( F(\phi)) = \psi(\phi^+).
\end{align}
Suppose that $\tilde \psi\in\mrm{Lip}(\mc Z)$ is induced by
$\psi\in\mrm{Lip}(W)$ via $\tilde \psi=\psi\circ\mc R$. Then,
the action of the induced Koopman operator $\tilde{\mc K}$ is given by
\begin{align*}
(\tilde{\mc K}\tilde \psi)(z)
=
\tilde \psi(\tilde F(z)). %= \psi\bigl(\mc P \circ F \circ \mc P(\phi)\bigr).
\end{align*}
For $\mbb X=\{z_\ell\}_{\ell=1}^p$ and each $\ell \in [1:p]$, define $\Psi_\ell:\mc Z\to\R$ by
$\Psi_\ell(\cdot):=\varphi_{z_\ell}(\cdot)$, and let
$
V_{\mbb X}:=\mrm{span}\{\Psi_1,\ldots,\Psi_p\}\subset \mc H.
$
Define the kernel interpolation operator
$\mc I_{\mbb X}:\mrm{Lip}(\mc Z) \! \to \!  V_{\mbb X}$ by
\[
(\mc I_{\mbb X}g)(\cdot)
:=
{\scriptstyle \sum}_{\ell=1}^p
\bigl(K_{\mbb X}^{-1}g_{\mbb X}\bigr)_\ell
\Psi_\ell(\cdot),
\  
g_{\mbb X}:=
[g(z_1)\cdots g(z_p)]^\top,
\]
for any $g\in\mrm{Lip}(\mc Z)$.
Since $V_{\mbb X}\subset\mc H\hookrightarrow\mrm{Lip}(\mc Z)$,
for any $g\in V_{\mbb X}$, the action of $\tilde{\mc K}$ can be
approximated by kernel interpolation as \cite{kohne2025error}
\eqnn{
(\tilde{\mc K}g)(z)
\approx
(\mc I_{\mbb X}\tilde{\mc K}g)(z)
=
{\scriptstyle \sum}_{\ell=1}^p
\bigl(\tilde K g_{\mbb X}\bigr)_\ell
\Psi_\ell(z),
\label{eq:koop_ev}
}
where $\tilde{ K} = K_{\mbb X}^{-1}K_{\tilde F(\mbb X)}K_{\mbb X}^{-1}$, and
\begin{align}
K_{\tilde F(\ds X)} =
\begin{bmatrix}
\mrm k(z_1,\tilde F(z_1)) & \cdots & \mrm k(z_p,\tilde F(z_1))\\
\vdots & \ddots & \vdots \\
\mrm k(z_1,\tilde F(z_p)) & \cdots & \mrm k(z_p,\tilde F(z_p))
\end{bmatrix}.
\label{eq:K_F}
\end{align}

\subsection{Data Driven Surrogate using kEDMD}\label{sec:kedmd}

Consider each data point $z_\ell \in \ds X$ and the corresponding local dataset $\ds X_\ell$. To estimate the nonlinear dynamics locally, we solve the affine regression problem
\eqnn{
\Phi_\ell
=
\argmin_{\Phi \in \R^{nM \times (nM+1)}}
\| \mf Z_\ell^+ - \Phi \bar{\mf Z}_\ell \|_F,
\label{eq:edmd_ls}
}
where $\Phi_\ell = [\hat F_\ell \ \ \hat B_\ell]$, with
$\hat F_\ell \in \R^{nM \times 1}$ and $\hat B_\ell \in \R^{nM \times nM}$.
The local regression \eqref{eq:edmd_ls} is used to estimate $\tilde F(z_\ell)$, with the discrepancy between the true successors $z_{\ell,j}^+$ and $\tilde F(z_{\ell,j})$ accounted for in the subsequent error analysis.
The regression problem \eqref{eq:edmd_ls} is well-posed under \Cref{assm3}. Using the estimates $\hat F_\ell \approx \tilde F(z_\ell)$ for each $z_\ell \in \ds X$, we construct a linear surrogate model based on kEDMD. Consider the lifted state
\eqnn{
\bs \Psi(z) = [\Psi_1(z) \ \cdots \ \Psi_p(z)]^\top,
}
where $\Psi_\ell(\cdot)=\varphi_{z_\ell}(\cdot)$ for $\ell \in [1:p]$. Then, from \eqref{eq:koop_ev}, the induced Koopman evolution admits the approximation
\eqnn{
(\tilde{\mc K}\bs\Psi)(z) \approx \mf A \bs\Psi(z), \text{where } \mf A = K_{\hat F(\ds X)}^\top K_{\ds X}^{-1},
\label{eq:koop_app}
}
and $K_{\hat F(\ds X)}$ is constructed using the estimates $\hat F_\ell$ of $\tilde F(z_\ell)$ for $\ell \in [1:p]$, analogously to \eqref{eq:K_F}. Using the Koopman approximation \eqref{eq:koop_app}, the lifted state evolution is approximated by the data-driven surrogate model
\eqnn{
\bs \Psi(z^+) \approx \mf A \bs \Psi(z).
\label{eq:surr1}
}
Next, we state the main theorem, which provides a bound on the approximation error of the data-driven linear surrogate \eqref{eq:surr1}.

\begin{shadeboxdark}
\begin{theorem}\label{thm1}
Let $M \ge 2$ be fixed, and suppose Assumptions \ref{assm1}, \ref{assm2}, and \ref{assm3} hold. For any sampled history state $\phi \! \! \in \!  \! W$ with $z\!=\! \mc Q(\phi)$ and exact successor $z^+ \! =\! \mc Q(F(\phi))$, the exact discretized evolution \eqref{eq:dde_disc} admits the lifted surrogate representation
 \eqnn{
 \bs \Psi(z^+) = \mf A \bs \Psi(z) + \mf r_\phi,
  \label{eq:surr}
 }
 where the residual satisfies 
 \eqnn{ \label{eq:det_bound}
 \| \mf r_\phi\| \leq  R_{\mrm{se}} \delta(M) + R_i h_{\mbb X} + R_d \ \rho^2,
 }
 with $R_{\mrm{se}} = R_s+R_e$, and
 \begin{subequations}
 \begin{align}
   R_s &= C_1 L_F L_\theta \|\bs \Psi\|_\mc{H},\
    R_i = B_\mbb{X} \|\bs \Psi \|_\mc{H},\\
    B_\mbb{X} & = C_1(L_F+ \sqrt{p} C_0 \| K_{\mbb X}^{-1}\| \| \bs \Psi\|_\mc{H})\\
    R_d &=  C_0 \sqrt{ d}\ p \ L_{\mrm k}\bar Z \tfrac{D_F}{2} \| \bs \Psi\|_\mc{H} \| K_{\ds X}^{-1} \|,\\
    R_e &=  C_0 \sqrt{ dM} p L_{\mrm k} L_F L_\theta\bar Z \|\bs \Psi\|_\mc{H} \| K_{\ds X}^{-1} \|. 
 \end{align}
 \end{subequations}
Here, $\|\bs\Psi\|_{\mc H}
:=
\left(
{\scriptstyle \sum}_{\ell=1}^p
\|\Psi_\ell\|_{\mc H}^2
\right)^{1/2}$, $\bar Z = \max_{\ell \in [1:p]} \| \bar{\mf Z}_\ell^\dagger \|$, $C_0$ and $C_1$ are the constants defined in the proof of \Cref{lem:rkhs_lip}, and $L_{\mrm k}>0$ depends only on the kernel function.
\end{theorem}
\end{shadeboxdark}

\begin{proof}
    See  Appendix (\Cref{sec:proofs}).
\end{proof}
% Equation \eqref{eq:det_bound} provides a deterministic error bound for the kEDMD-based surrogate model \eqref{eq:surr}, decomposed into three terms: the state discretization error \(R_s\delta(M)\), the kernel interpolation error \(R_i(\|z\| + d_{\ds X})\), and the data-driven regression error \(R_d\rho^2\). The discretization term vanishes as \(M \to \infty\). If \(0 \in \ds X\), then \(\operatorname{dist}(z,\ds X) \le \|z\|\), so the interpolation error scales proportionally with \(\|z\|\). Since \(R_i\) depends on the fill distance \(h_{\ds X}\), we have \(R_i \to 0\) as \(h_{\ds X} \to 0\); in particular, under sufficiently dense sampling of \(\ds X\) over \(\mc Z\), this implies \(R_i \to 0\) as \(p \to \infty\). The regression error term decreases quadratically as \(\rho \to 0\), provided \Cref{assm3} remains satisfied. However, because \(\|K_{\ds X}^{-1}\|\) generally increases with \(p\) \cite{wendland2004scattered}, \(\rho\) must be chosen small enough to reduce regression error while preserving numerical conditioning.

Equation \eqref{eq:det_bound} provides a deterministic error bound for the
kEDMD-based surrogate \eqref{eq:surr}, comprising the history discretization
error $R_{\mrm{se}}\delta(M)$, kernel interpolation error $R_i h_{\ds X}$,
and local regression error $R_d\rho^2$. Increasing $M$ decreases the grid
spacing $\delta(M)$, while increasing $p$ can reduce $h_{\ds X}$ through
denser coverage of $\mc Z$. The regression error is quadratic in $\rho$,
provided \Cref{assm3} remains satisfied. Since $\|K_{\ds X}^{-1}\|$ may
increase with $p$ \cite{wendland2004scattered}, denser sampling can improve
interpolation while worsening the conditioning of $K_{\ds X}$.

\begin{shadeboxdark}
\begin{proposition}\label{prop1}
For any $z \in \mc Z$, the discretized state reconstructed from the lifted coordinates via kernel interpolation is given by
$
\hat z = \mf C K_{\ds X}^{-1}\bs \Psi(z),
$
where \(\mf C = [z_1\ \cdots\ z_p] \in \R^{nM \times p}\) with \(z_\ell \in \ds X\) for all \(\ell \in [1:p]\). Consequently, the one-step predicted state under the
data-driven surrogate model is
$
\hat z^+ = \mf C K_{\ds X}^{-1}\mf A \bs \Psi(z).
$
Moreover,
\eqnn{
 \|z^+ - \hat z^+\| \le B_p h_\ds{X}+ \|\mf C K^{-1}_\ds{X}\|\|\mf r_\phi\|, \label{eq:prop3}
}
where $B_p =  1+C_1 \|\mf C K_{\ds X}^{-1} \| \| \bs \Psi\|_{\mc H}$. 
\end{proposition}
\end{shadeboxdark}
\begin{proof}
    See Appendix (\Cref{sec:proofs}).
\end{proof}

Proposition~1 shows that the discretized state can be reconstructed from the lifted coordinates via kernel interpolation, yielding a corresponding one-step state prediction. The bound \eqref{eq:prop3} separates the resulting error into the kernel reconstruction error and the lifted-surrogate residual propagated through $\mf C K_{\ds X}^{-1}$. Thus, prediction accuracy is governed jointly by kernel reconstruction and the residual bound in \Cref{thm1}.

\begin{remark}
Over the admissible horizon, the one-step lifted-surrogate bound in \Cref{thm1} extends to multiple steps, with residuals propagated through powers of $\mf A$. Likewise, the state-prediction bound in \Cref{prop1} extends to multiple steps, combining the accumulated lifted-surrogate error with the kernel-based reconstruction error.
\end{remark}

\section{Numerical Results}

To numerically validate the proposed Koopman linear surrogate model, we consider two systems: a scalar DDE with Hill-type nonlinearity \cite{glass2021nonlinear} ($\tau_d=1$) 
\eqnn{
\dot x(t) = \tfrac{1}{1+\left(x(t-\tau_d) \right)^2} -x(t),
\label{eq:hill}
}
and a tumor--immune interaction model \cite{rihan2014delay} ($\tau_d=1.636$)
\begin{subequations}\label{eq:tum}
\begin{align}
& \dot{x}_1(t)
=
0.0441
+ 0.6913\,\tfrac{x_1(t-\tau_d)\,x_2(t-\tau_d)}{1+x_2(t-\tau_d)} \notag
  - \\
& \qquad \quad 0.0384\,x_1(t-\tau_d)\,x_2(t-\tau_d)
- 0.2288\,x_1(t),\\
& \dot{x}_2(t)
=
 x_2(t)\bigl(1 - 0.04038\,x_2(t)\bigr)
- x_1(t)\,x_2(t).
\end{align}
\end{subequations}

In the numerical experiments, explicitly constructing the full theoretical domain $\mc Z$ from a priori bounds is generally impractical. We therefore simulate the DDEs using \texttt{JiTCDDE} \cite{ansmann2018efficiently}  over a finite horizon from an ensemble of constant initial history functions, construct the corresponding history states $x_t$, and sample them on $\Theta_M$ to obtain the discretized states $z(t)=\mc Q(x_t)$. These states define the empirically explored portion of $\mc Z$, over which numerical quantities such as the fill distance are evaluated. This empirical construction is used only for computation and does not alter the theoretical definition of $\mc Z$ in Section~II-B. We compute $h_\ds{X}$ for this empirical domain for a particular $p$ using the procedure in \cite{climaco2024on}.  The dataset $\mathcal T$ is constructed from the samples used to form $\mc Z$ and we use kEDMD procedure in \Cref{sec:kedmd} to form the surrogate model. We also draw 100 samples of independent test trajectory data (different from the ones used in the learning) from the samples forming $\mc Z$. 
Along each test trajectory $\{z_k\}_{k=0}^{N_t}$, we construct the predicted trajectory $\{\hat z_k\}_{k=0}^{N_t}$, initialized with $\hat z_0 = C K_{\ds X}^{-1}\bs \Psi(z_0)$. The prediction error at time step $k \in [1:N_t]$ for the $i$-th test trajectory is defined as $e_k^{(i)} = \|z_k^{(i)} - \hat z_k^{(i)}\|$. The step-wise \textit{mean prediction error} (MPE) on the $100$ test trajectories is then given by $\mu_{z_k} = \tfrac{1}{100} \sum_{i=1}^{100} e_k^{(i)}$. We used $\rho=0.6$ unless otherwise noted.

For the system \eqref{eq:hill}, \Cref{fig:one}(a) shows the evolution of $\mu_{z_k}$ along the test trajectories for different discretization levels $M$, with $p$ chosen to yield comparable fill distances, demonstrating consistently lower $\mu_{z_k}$ as $M$ increases ($h_\ds{X} \approx 0.0053$). Fixing $M=3$, \Cref{fig:one}(b) shows the variation of $\mu_{z_k}$ with $p$, with $\mu_{z_k}$ decreasing as $p$ increases, while \Cref{fig:one}(c) illustrates its dependence on $\rho$, where smaller $\rho$ further reduces $\mu_{z_k}$ for $p=500$. In \Cref{fig:two}, we compare the true current value of the state function, corresponding to the last component of the discretized state $z$ (i.e., at $\theta_M$), with the corresponding prediction obtained from the last component of $\hat z$. For the planar system \eqref{eq:tum}, \Cref{fig:two}(b)--(c) show close agreement between the true and predicted trajectories of $x_1(t)$ and $x_2(t)$ for $p=1500$ and $M=3$. Similarly, \Cref{fig:two}(a) shows close agreement for \eqref{eq:hill} with $p=100$ and $M=3$. These results demonstrate the effectiveness of the proposed kEDMD-based surrogate model in accurately reconstructing the discretized state and predicting the current state from the lifted coordinates.

% ---- end sections/4_numerical ----

\begin{figure*}[htpb!]
\centering
\subfloat[]{\includegraphics[trim={0 0 0 0}, clip=true, width=0.28\textwidth]{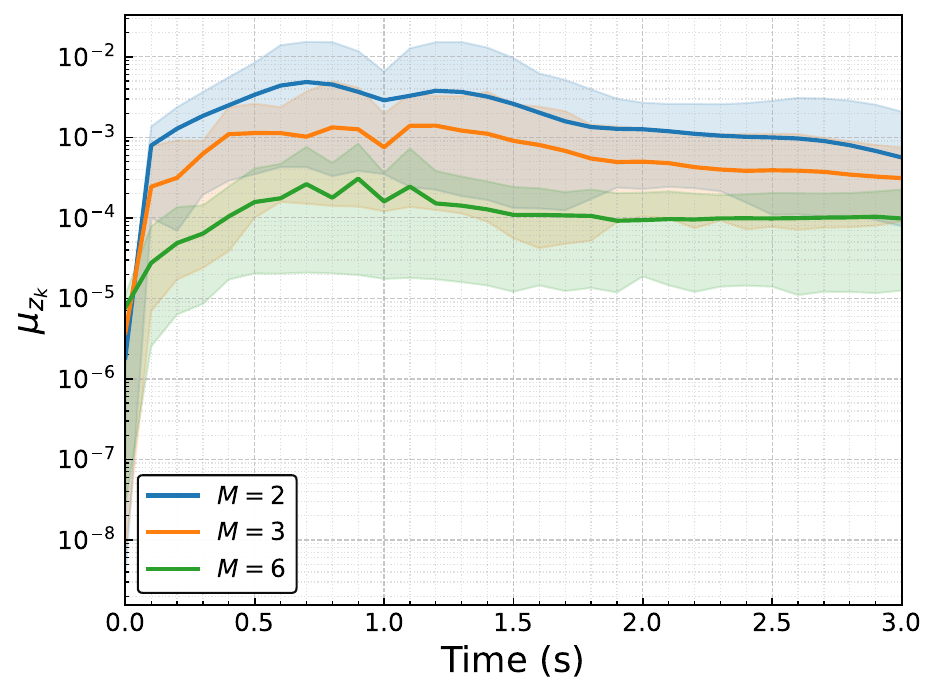} \label{fig:pred}}
\subfloat[]{\includegraphics[trim={0 0 0 0}, clip=true, width=0.28\textwidth]{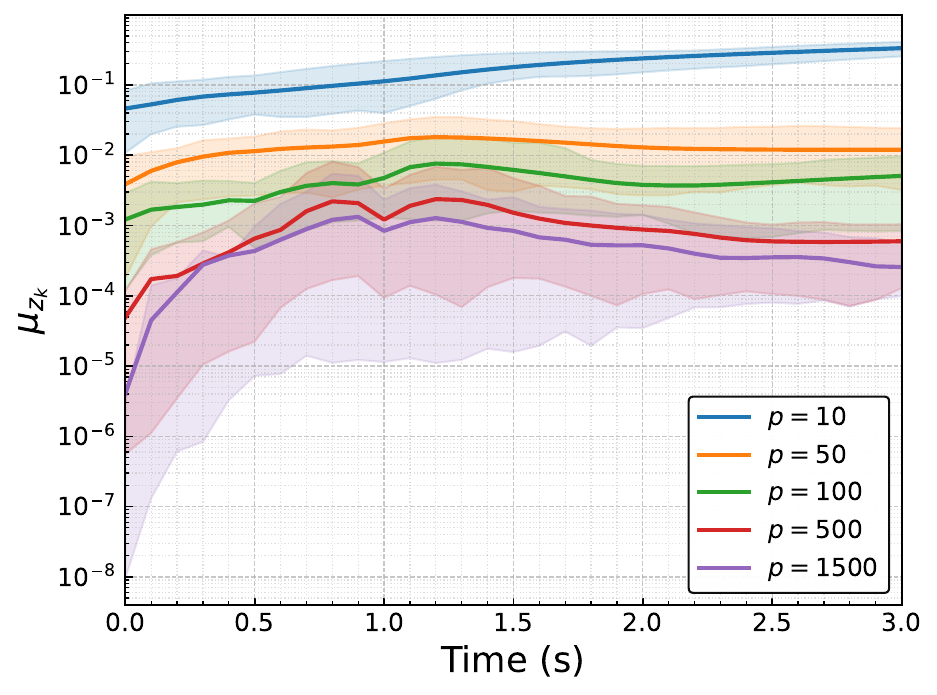} \label{fig:track}}
\subfloat[]{\includegraphics[trim={0 0 0 0}, clip=true, width=0.28\textwidth]{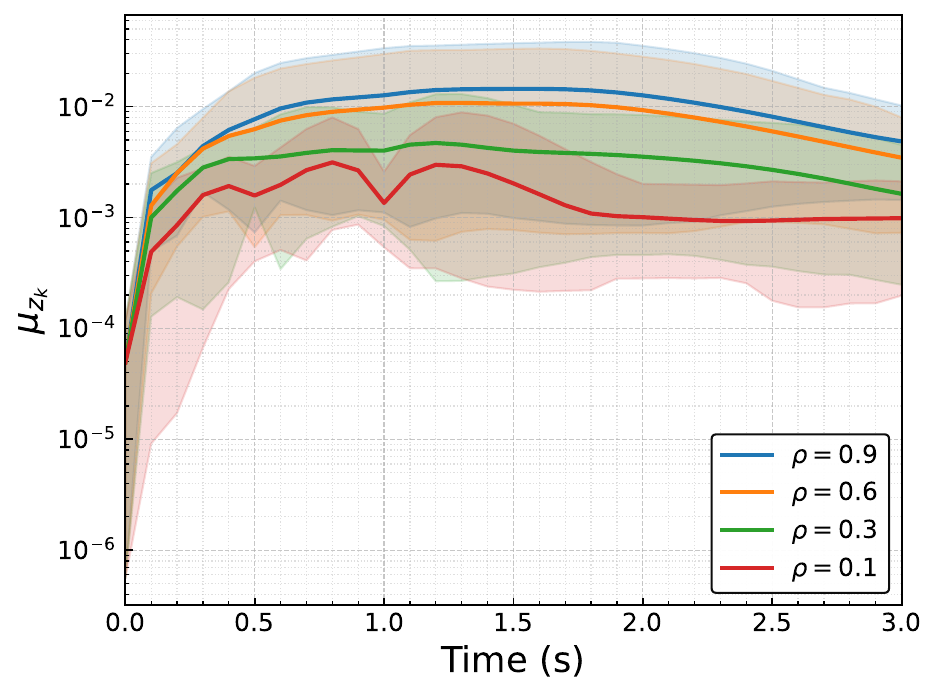} \label{fig:track}}
%\vspace{-1cm}
\caption{Step-wise MPE $\mu_{z_k}$ for \eqref{eq:hill}: (a) increasing $M$, (b) increasing $p$, (c) decreasing $\rho$; shaded area depicts the range from worst to best error.}
\label{fig:one}  \vspace{-1em}
\end{figure*}

\begin{figure*}[htpb!]
\centering
\subfloat[(pooled RMSE = 0.0026)]{\includegraphics[trim={0 0 0 0}, clip=true, width=0.28\textwidth]{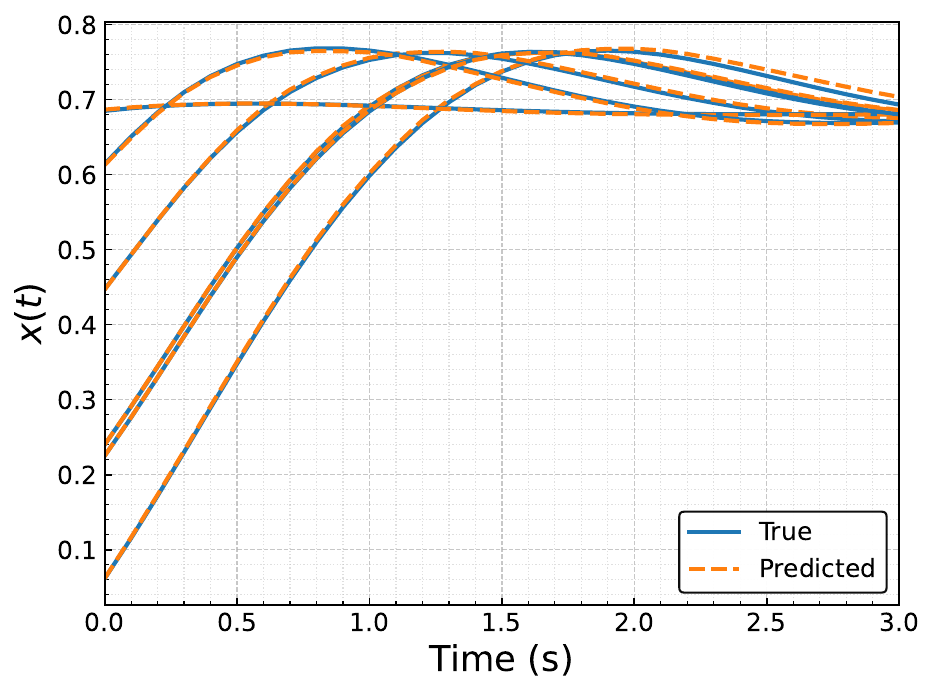} \label{fig:pred}}
\subfloat[(pooled RMSE = 0.025)]{\includegraphics[trim={0 0 0 0}, clip=true, width=0.28\textwidth]{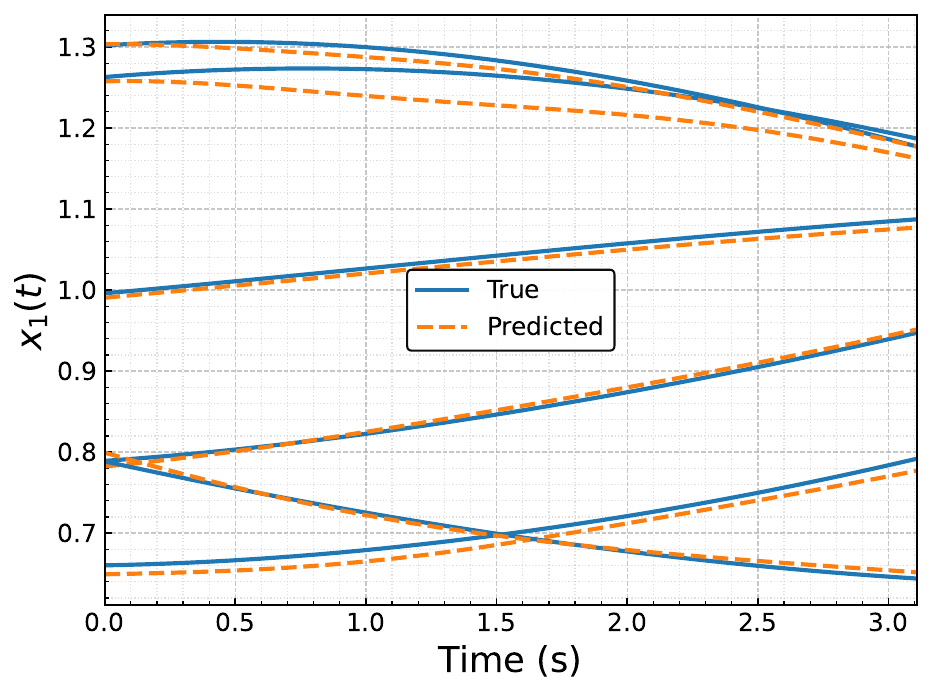} \label{fig:track}}
\subfloat[(pooled RMSE = 0.16)]{\includegraphics[trim={0 0 0 0}, clip=true, width=0.28\textwidth]{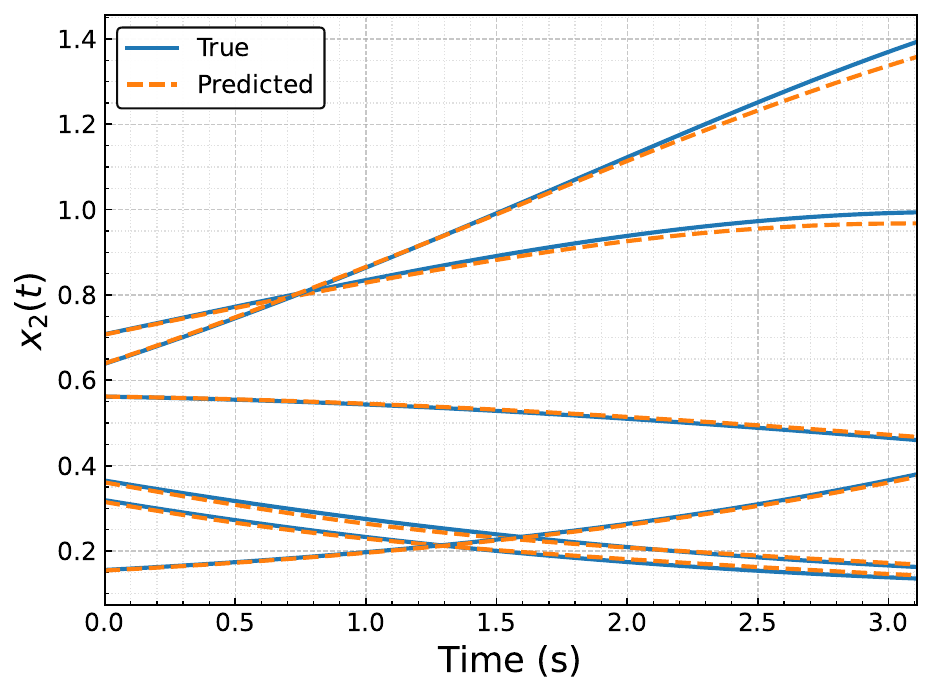} \label{fig:track}}
%\vspace{-1cm}
\caption{Evolution of current value of the history function: (a) $x(t)$ for system \eqref{eq:hill}, (b) $x_1(t)$ for system \eqref{eq:tum}, (c) $x_2(t)$ for system \eqref{eq:tum}.}
\label{fig:two}  \vspace{-1em}
\end{figure*}

\section{Conclusion}

This work develops a data-driven Koopman framework for nonlinear delay differential equations (DDEs), bridging infinite-dimensional delay dynamics with finite-dimensional representations via history discretization and reconstruction. Using kernel-based extended dynamic mode decomposition (kEDMD), we construct a linear surrogate model with deterministic error bounds that decompose the approximation error into discretization, kernel interpolation, and regression components, along with guarantees for state reconstruction. Numerical results demonstrate convergence with respect to discretization resolution and data density. Overall, the proposed framework provides a theoretically grounded foundation for learning and prediction of delay systems, with future directions including closed-loop control synthesis and extensions to high-dimensional networked systems.

% ---- begin sections/appendix ----

\section{Appendix}\label{appen} \label{sec:proofs}

\noindent \textbf{Proof of \Cref{lem1}:} \textup{(i)} Let $\zeta \in \mathcal{Z}$ and the reconstructed state function be $\tilde \eta = \mathcal{R}(\zeta)$. It follows from \eqref{eq:recon} that for any $\theta \in [\theta_j, \theta_{j+1}]$, $\tilde \eta (\theta)$ is a convex combination of $\zeta_j$ and $\zeta_{j+1}$. Since $\|\zeta_j\| \leq \gamma$ for all $j$, by convexity of the norm, $\|\tilde \eta(\theta)\| \leq \gamma$. Thus, $\sup_{\theta \in [-\tau_d,0]}\|\tilde \eta(\theta)\| = \| \tilde \eta \|_\infty \leq \gamma$. Moreover, on each interval $[\theta_j,\theta_{j+1}]$, $\tilde \eta$ is affine in $\theta$ and satisfies
$
\left\|\tfrac{\dx}{\dx\theta}\tilde \eta(\theta)\right\|
=
\tfrac{\|\zeta_{j+1}-\zeta_j\|}{\delta(M)}
\leq L_\theta.
$
Since $\tilde \eta$ is continuous and piecewise differentiable in $\theta$, it follows that
$
\|\tilde \eta(\theta_1)-\tilde \eta(\theta_2)\|
\leq
L_\theta |\theta_1-\theta_2|,
\qquad \forall \theta_1,\theta_2 \in [-\tau_d,0].
$
Therefore, $\tilde \eta \in W$. Since $\zeta \in \mathcal{Z}$ was arbitrary, we conclude that $\mathcal{R}(\mathcal{Z}) \subseteq W$.

Now let $\zeta,\xi\in\mc Z$. For any $\theta\in[\theta_j,\theta_{j+1}]$, write $\lambda=(\theta-\theta_j)/\delta(M)\in[0,1]$. Then
\[
\mc R(\zeta)(\theta)-\mc R(\xi)(\theta)
=
(1-\lambda)(\zeta_j-\xi_j)+\lambda(\zeta_{j+1}-\xi_{j+1}).
\]
By convexity of the norm,
\begin{align*}
 \|\mc R(\zeta)(\theta)-\mc R(\xi)(\theta)\|
{\scriptstyle \le
 (1-\lambda)\|\zeta_j-\xi_j\|+\lambda\|\zeta_{j+1}-\xi_{j+1}\|  \le 
\|\zeta-\xi\|_{b,\infty} }.
\end{align*}
Taking the supremum over $\theta\in[-\tau_d,0]$ yields
$
\|\mc R(\zeta)-\mc R(\xi)\|_\infty \le \|\zeta-\xi\|_{b,\infty}.
$

\textup{(ii)} By definition, $\mc Q(\mc R(\zeta))$ samples the function $\mc R(\zeta)$ at the grid points $\{\theta_j\}_{j=1}^M$. From the interpolation formula \eqref{eq:recon}, $\mc R(\zeta)(\theta_j) = \zeta_j$ for all $j \in[1:M]$. Thus, $\mc Q(\mc R(\zeta)) = \zeta$. Hence, the composition $\mc Q \circ \mc R$ is the identity on $\mc Z$.

% \textup{(iii)} Let $\eta \in W$. On any interval $[\theta_j, \theta_{j+1}]$, $\mc P(\eta)$ is the linear interpolant of $\eta$ at the endpoints. Since $\eta$ is Lipschitz continuous with constant $L_\theta$ (by Assumption \ref{assm1}), the maximum deviation between $\eta$ and its linear interpolant on an interval of width $\delta(M)$ is bounded by $L_\theta \delta(M)$. This follows from the standard theory of local polynomial reproduction, where the error scales with the fill distance \cite[Chapter 3]{wendland2004scattered}. Hence, $\|\eta - \mc P(\eta)\|_\infty \leq L_\theta \delta(M), \quad \forall \eta \in W.$

\textup{(iii)} By \textup{(i)}, $\mc P(W)\subseteq W$, and by \textup{(ii)},
$\mc P^2=\mc R\circ\mc Q\circ\mc R\circ\mc Q
=\mc R\circ\mc Q=\mc P$; hence, $\mc P$ is a projection on $W$. Let $\eta \in W$. For any $\theta \in [\theta_j,\theta_{j+1}]$, define $\lambda := (\theta-\theta_j)/\delta(M) \in [0,1]$. Then,
$
\mc P(\eta)(\theta)
=
(1-\lambda)\eta(\theta_j)+\lambda \eta(\theta_{j+1}).
$
Hence,
\begin{align*}
\|\eta(\theta)-\mc P(\eta)(\theta)\| &= {\scriptstyle
\|(1-\lambda)\bigl(\eta(\theta)-\eta(\theta_j)\bigr)
+  \lambda \bigl(\eta(\theta)-\eta(\theta_{j+1})\bigr)\|} \\
&\le
{\scriptstyle (1-\lambda)\|\eta(\theta)-\eta(\theta_j)\|
+ \lambda \|\eta(\theta)-\eta(\theta_{j+1})\|} \\
&\le
{\scriptstyle (1-\lambda)L_\theta |\theta-\theta_j|
+  \lambda L_\theta |\theta-\theta_{j+1}|}.
\end{align*}
Using $|\theta-\theta_j|=\lambda \delta(M)$ and $|\theta-\theta_{j+1}|=(1-\lambda)\delta(M)$, we obtain
$
\|\eta(\theta)-\mc P(\eta)(\theta)\|
\le
2\lambda(1-\lambda)L_\theta \delta(M)
\le
L_\theta \delta(M).
$
Taking the supremum over $\theta \in [-\tau_d,0]$ yields
$
\|\eta - \mc P(\eta)\|_\infty \le L_\theta \delta(M),
$
which completes the proof. \hfill $\blacksquare$

\vspace{0.5em}

\noindent \textbf{Proof of \Cref{lem:tildeF_hessian}:} 

Let $U\subset\mathscr C$ be an open neighborhood of $W$ on which $F$ is twice continuously Fr\'echet differentiable, and let $\overline{\mc Q}:\mathscr C\to\mathbb R^{nM}$ and $\overline{\mc R}:\mathbb R^{nM}\to\mathscr C$ denote the bounded linear extensions of $\mc Q$ and $\mc R$, respectively. Since $\overline{\mc R}(\mc Z)=\mc R(\mc Z)\subseteq W\subset U$ by Lemma~\ref{lem1}, the set $V:=\overline{\mc R}^{-1}(U)$ is an open neighborhood of $\mc Z$. Hence, $\overline{\widetilde F}:=\overline{\mc Q}\circ F\circ\overline{\mc R}:V\to\mathbb R^{nM}$ is twice continuously Fr\'echet differentiable. Moreover, $\mc R(\mc Z)\subseteq W$, $F(W)\subseteq W$ by Assumption~\ref{assm1} and $\Delta\le t_f$, and $\mc Q(W)\subseteq\mc Z$; therefore, the restriction of $\overline{\widetilde F}$ to $\mc Z$ coincides with $\widetilde F:\mc Z\to\mc Z$.

Since $\mc Z$ is closed and convex with nonempty interior, $\mc Z=\overline{\operatorname{int}\mc Z}$. Hence, the second derivative of any twice continuously Fr\'echet differentiable extension of $\widetilde F$ is uniquely determined on $\mc Z$: any two such extensions have identical second derivatives on $\operatorname{int}\mc Z$, and therefore on $\mc Z$ by continuity. We may thus write $D^2\widetilde F(z):=D^2\overline{\widetilde F}(z)$ for $z\in\mc Z$.

For any $z\in\mc Z$ and $h_1,h_2\in\mathbb R^{nM}$, the second-order chain rule and the linearity of $\overline{\mc Q}$ and $\overline{\mc R}$ yield $D^2\widetilde F(z)[h_1,h_2]=\overline{\mc Q}\!\left(D^2F(\mc R(z))[\overline{\mc R}(h_1),\overline{\mc R}(h_2)]\right)$. Furthermore, $\|\overline{\mc Q}(\eta)\|\le\sqrt{M}\|\eta\|_\infty$ for every $\eta\in\mathscr C$, while the interpolation formula defining $\overline{\mc R}$ gives $\|\overline{\mc R}(h)\|_\infty\le\|h\|_{b,\infty}\le\|h\|$ for every $h\in\mathbb R^{nM}$. Therefore, Assumption~\ref{assm2} implies $\|D^2\widetilde F(z)[h_1,h_2]\|\le\sqrt{M}\,\bar D_F\|h_1\|\|h_2\|$. Taking the supremum over $\|h_1\|\le1$ and $\|h_2\|\le1$ gives $\|D^2\widetilde F(z)\|_{\mathcal L^2(\mathbb R^{nM},\mathbb R^{nM})}\le\sqrt{M}\,\bar D_F=D_F$. \hfill $\blacksquare$

\vspace{0.5em}

\noindent \textbf{Proof of Theorem 1:}
Let $\what{\bs \Psi}(\cdot) := \bs \Psi \circ \mc Q(\cdot) = [ \Psi_1 \circ \mc Q(\cdot) \ \cdots \  \Psi_p \circ \mc Q(\cdot)]^\top $. We want a rigorous error bound for 
\begin{align}
   \!\!\! &\bs \Psi(z^+)-\mf A\bs \Psi(z) = \bs \Psi(\mc Q(F(\phi)))- \mf A\bs \Psi(z) \notag \\
    \!\!\! &= \mc K \what{\bs \Psi}(\phi) - \tilde{\mc K} \bs \Psi(z) + \tilde{\mc K} \bs \Psi(z)-\mf A\bs \Psi(z) \notag \\
    \!\!\! & = \what{\bs \Psi}(F(\phi)) - \bs \Psi(\tilde F(z)) + \tilde{\mc K} \bs \Psi(z)-\mf A\bs \Psi(z) \label{eq:err_decomp}\\
   \!\!\!  &= \what{\bs \Psi}(F(\phi)) - \what{\bs \Psi}(F \circ \mc P(\phi)) + \tilde{\mc K} \bs \Psi(z)- \mc I_{\ds X}\tilde{\mc K} \bs \Psi(z) \notag \\
   \!\!\!  & \qquad \qquad \qquad \qquad \qquad + \mc I_{\ds X}\tilde{\mc K} \bs \Psi(z) - \mf A\bs \Psi(z).\notag
\end{align}

For each $\ell \in [1:p]$, define $g_\ell(\cdot) := \tilde{\mc K}\Psi_\ell(\cdot) = \Psi_\ell \circ \tilde F(\cdot)$. For any $z_1,z_2 \in \mc Z$, using \Cref{lem1}, we obtain $\| \tilde F(z_1)-\tilde F (z_2) \|_{b,\infty} \le L_F \| z_1-z_2\|_{b,\infty}$. Since $\Psi_\ell \in \mc H$, \Cref{lem:rkhs_lip} gives $\mrm{Lip}_\mc{Z}(\Psi_\ell) \le C_1 \| \Psi_\ell \|_\mc{H}$. Therefore, $|g_\ell(z_1)-g_\ell(z_2)| \le C_1L_F \| \Psi_\ell \|_\mc{H}\| z_1-z_2\|_{b,\infty}$. Hence, $\mrm{Lip}_\mc{Z}(g_\ell) \le C_1 L_F \| \Psi_\ell\|_\mc{H}$. Now, define
\begin{align}
    \alpha^{(\ell)} := K_{\ds X}^{-1} (g_\ell)_{\ds X}, \ (g_\ell)_{\ds X} := [g_\ell(z_1)\ \cdots \ g_\ell(z_p)]^\top.
    \label{eq:alpha_ell}
\end{align}
Using results from \Cref{lem:rkhs_lip}, we obtain $\| (g_\ell)_{\ds X}\| \le \sqrt{p} C_0 \| \Psi_\ell\|_\mc{H}$, and consequently $\| \alpha^{(\ell)}\| \le \sqrt{p} C_0 \| K_{\ds X}^{-1}\| \| \Psi_\ell\|_\mc{H}$. By definition, $\mc I_{\ds X} g_\ell = {\scriptstyle \sum}_{m=1}^p \alpha^{(\ell)}_m \Psi_m$, giving
\begin{align*}
    \mrm{Lip}_\mc{Z} (\mc I_{\ds X} g_\ell) & \le C_1 \| \alpha^{(\ell)}\| \| \bs \Psi\|_\mc{H} \\
    & \le C_1 \sqrt{p} C_0 \| K_{\ds X}^{-1}\| \| \Psi_\ell\|_\mc{H} \| \bs \Psi\|_\mc{H}.
\end{align*}
Therefore,
$
\mrm{Lip}_\mc{Z} (g_\ell-\mc I_{\ds X} g_\ell) \le B_\ds{X} \| \Psi_\ell\|_\mc{H},
$
where $B_\ds{X} = C_1(L_F+ \sqrt{p} C_0 \| K_{\ds X}^{-1}\| \| \bs \Psi\|_\mc{H})$. Hence,
\begin{align}
    \| \tilde{\mc K} \bs \Psi(z) - \mc I_{\ds X} \tilde{\mc K} \bs \Psi(z) \| & \le B_\ds{X} \mrm{dist}(z, \ds X) \| \bs \Psi \|_\mc{H} \notag\\
    & \le B_\ds{X} \| \bs \Psi \|_\mc{H} h_\ds{X}
    \label{eq:ri}
\end{align}

%  Note that $\mf A$ can be thought of as a perturbed matrix representation of $\mc I_{\ds X}\tilde{\mc K}$ in the canonical basis given by $\{ \Psi_\ell \}_{\ell=1}^p$, with $\mc I_{\ds X}$ is the orthogonal projection onto $\mrm{span}\{ \Psi_1, \cdots, \Psi_p \}$.  Using results in \cite[Section 3.2]{bold2025kernel}, \cite[Section 3]{wendland2004scattered} and \cite{strasser2025kernel}, there exists $B_1, h_0 >0$ such that for $\ell \in[1:p]$ and $ h_\ds{X} \le h_0$
% \begin{align}
% |\tilde{\mc K}  \Psi_\ell(z)- \mc I_{\ds X}\tilde{\mc K}\Psi_\ell(z)| \leq B_1 h_\ds{X}^{s-\tfrac{1}{2}} \mrm{dist}(z, \ds X) \|\Psi_\ell \|_\mc{H}. 
% \end{align}
% The constants $B_1$ and $h_0$ depend only on the domain $\mc Z$ \cite{wendland2004scattered}. Also, $\mrm{dist}(z,\ds X) \le \| z\| + d_{\ds X}$. Therefore,
% \begin{align}
% |\tilde{\mc K}  \Psi_\ell(z)- \mc I_{\ds X}\tilde{\mc K}\Psi_\ell(z)| \leq B_1 h_\ds{X}^{s-\tfrac{1}{2}} \|\Psi_\ell \|_\mc{H} (\| z\|+d_\ds{X}).
% \end{align} 
% With $\| \bs \Psi \|_\mc{H}$ as in the definition \cite[c.f., Sec III]{strasser2025kernel}, we obtain 
% \begin{align}
% \| \tilde{\mc K} \bs \Psi(z)- \mc I_{\ds X}\tilde{\mc K} \bs \Psi(z) \| \le  B_1 h_\ds{X}^{s-\tfrac{1}{2}} \|\bs \Psi \|_\mc{H} (\| z\|+d_\ds{X}).
% \label{eq:ri}
% \end{align}
Note that $\mc I_\ds{X} \tilde{\mc K}\bs \Psi(z) \! = \!  K^\top_{\tilde F(\ds X)} K_{\ds X}^{-1} \bs \Psi(z)$. Let $\mf A^\star = K^\top_{\tilde F(\ds X)} K_{\ds X}^{-1}$. Then 
\begin{align}
   &\| \mc I_{\ds X}\tilde{\mc K} \bs \Psi(z) - \mf A\bs \Psi(z) \| = \|(\mf A^\star-\mf A)\bs \Psi(z)\|, \notag\\
   & \qquad \qquad  \le C_0 \|\bs \Psi \|_\mc{H} \| K^\top_{\tilde F(\ds X)}-K^\top_{\hat F(\ds X)}\| \| K_{\ds X}^{-1} \|. \label{eq: fin_err} 
\end{align}
For $j\in[1:d]$, define $e_{\ell,j}:= z^+_{\ell,j}- \tilde F(z_{\ell,j})$. Since $z_{\ell,j}= \mc Q(\phi_{\ell,j})$ for the corresponding true history $\phi_{\ell,j}$, $e_{\ell,j} = \mc Q \circ F(\phi_{\ell,j})- \mc Q \circ F\circ \mc P (\phi_{\ell,j})$. Therefore, using \Cref{assm2} and \Cref{lem1},
\begin{align}
\| e_{\ell,j} \| {\scriptstyle \le \sqrt{M} L_F \| \phi_{\ell,j}- \mc P(\phi_{\ell,j}) \|_\infty \le \sqrt{M} L_FL_\theta\delta(M)}.
\label{eq:e_lj}
\end{align}
Now, the Taylor expansion of $\tilde F$ about $z_\ell$, evaluated at a neighboring point $z_{\ell,j}$ yields
\begin{align}
    \tilde F(z_{\ell,j})
    =
    \tilde F(z_\ell)
    +
    J_{\tilde F}(z_\ell)\,\delta_{\ell,j}
    + r_{\ell,j},
\end{align}
where $\delta_{\ell,j} = z_{\ell,j}-z_\ell$, and
\begin{align*}
    r_{\ell,j} = {\scriptstyle \int_0^1}(1-s) D^2 \tilde F(z_\ell+s \delta_{\ell,j})[\delta_{\ell,j},\ \delta_{\ell,j}]ds.
\end{align*}
Because $\mathcal Z$ is convex, $z_\ell+s\delta_{\ell,j}\in\mathcal Z$ for all $s\in[0,1]$, so  \Cref{lem:tildeF_hessian} applies along the entire segment. Consequently,
\begin{align}
    \|r_{\ell,j} \| \le \tfrac{D_F}{2} \| \delta_{\ell,j}\|^2.
    \label{eq:r_lj}
\end{align}

Let $\Phi^\star_\ell = [\tilde F(z_\ell) \ \ J_{\tilde F}(z_\ell)]$ and $\Phi_\ell = \mf Z^+_\ell \bar{\mf Z}_\ell^\dagger$ be the solution to \eqref{eq:edmd_ls}. Let $\mf R_{\ell} = [r_{\ell,1} \ \cdots r_{\ell,d}]$ and $\mf E_\ell = [e_{\ell,1} \cdots e_{\ell,d}]$. Then 
\eqnn{
\mf Z^+_\ell = \Phi^\star_\ell\bar{\mf Z}_\ell + \mf R_{\ell}+\mf E_\ell.  \text{ Hence, }
}
\eqnn{
\Phi_\ell = (\Phi^\star_\ell\bar{\mf Z}_\ell + \mf R_{\ell}+\mf E_\ell)\bar{\mf Z}_\ell^\dagger = \Phi^\star_\ell + (\mf R_{\ell} + \mf E_\ell)\bar{\mf Z}_\ell^\dagger,
}
which gives
\eqnn{\| \tilde F(z_\ell) - \hat{F}_\ell\| = \|(\Phi_\ell-\Phi^\star_\ell) e_1\| \le \|\mf R_\ell + \mf E_\ell\| \| \bar{\mf Z}_\ell^\dagger \|,
}
where $e_1 = [1 \ \mf 0_{1\times nM}]^\top$. Let
$
\max_{j\in[1:d]} \|\delta_{\ell,j}\| \le \rho
$. Then, it follows from \eqref{eq:e_lj} and \eqref{eq:r_lj} that
\begin{align}
    \| \tilde F(z_\ell) - \hat{F}_\ell\| \le \sqrt{d}\bar Z \big( \tfrac{D_F}{2} \rho^2 + \sqrt{M}L_F L_\theta\delta(M) \big).
\end{align}
Since the Wendland kernel $\mrm k(\cdot,\cdot)$ has a continuously differentiable compactly supported radial profile, it is globally Lipschitz. Hence, there exists $L_{\mrm k} >0$, depending only on the kernel, such that
\[
|\mathrm{k}(\zeta,y_1) - \mathrm{k}(\zeta,y_2)|
\le
L_{\mrm k} \|y_1 - y_2\|,
\quad \forall \zeta \in \mc Z,\; \forall y_1,y_2 \in \R^{nM}.
\]
Let 
$\Delta K \coloneqq K^\top_{\tilde F(\ds X)} - K^\top_{\hat F(\ds X)}$ and denote by $\Delta K(:,\ell)$ its $\ell$-th column. 
Then, for all $\ell \in [1:p]$, we have
\begin{align}
\|\Delta K(:,\ell)\| \le \sqrt{dp}  L_{\mrm k} \bar Z \big( \tfrac{D_F}{2} \rho^2 + \sqrt{M}L_F L_\theta\delta(M) \big),
\end{align}
and with  $\| \Delta K \| \le \| \Delta K \|_F$, we have
\begin{align}
    \| \Delta K \| \le \sqrt{d} p L_{\mrm k} \bar Z \big( \tfrac{D_F}{2} \rho^2 + \sqrt{M}L_F L_\theta\delta(M) \big).
\end{align}
Therefore, it follows from \eqref{eq: fin_err} that
\begin{align}
  \| \mc I_{\ds X}\tilde{\mc K} \bs \Psi(z) - \mf A\bs \Psi(z) \|  \le R_d \rho^2 + R_e \delta(M),
  \label{eq:rd}
\end{align}
where $R_d=  C_0 \sqrt{ d} p  L_{\mrm k}\bar Z \tfrac{D_F}{2} \| \bs \Psi\|_\mc{H} \| K_{\ds X}^{-1} \| $ and $R_e =  C_0 \sqrt{ dM} p L_{\mrm k} L_F L_\theta\bar Z \|\bs \Psi\|_\mc{H} \| K_{\ds X}^{-1} \| $. Now, for each $\ell\in[1:p]$, using results from \Cref{lem:rkhs_lip}, we obtain
\begin{align*}
&\big|\widehat{\Psi}_\ell(F(\phi)) - \widehat{\Psi}_\ell(F\circ \mc P(\phi))\big| %\mrm{Lip}_\mc{Z}(\what \Psi)| F(\phi)-F \circ \mc P(\phi)|\\
 \le C_1\,\|\Psi_\ell\|_{\mathcal H} L_F \| \phi-\mc P(\phi) \|_\infty\\
& \qquad \qquad \qquad \qquad \qquad \le C_1\,\|\Psi_\ell\|_{\mathcal H}\,L_F\,L_\theta\,\delta(M).
\end{align*} 
Therefore, 
\begin{align}
\| \what{\bs \Psi}(F(\phi)) - \what{\bs \Psi}(F \circ \mc P(\phi)) \| \le R_s \,\delta(M),
\label{eq:rs}
\end{align}
where $R_s = C_1\ L_F\,L_\theta\,\|\bs \Psi\|_{\mathcal H}$. With $\mf r_\phi \! := \!  \bs \Psi(z^+) - \mf A \bs \Psi(z)$ and 
$R_i = B_\ds{X} \|\bs \Psi\|_{\mc{H}}$, 
it follows from \eqref{eq:err_decomp}, together with \eqref{eq:ri}, \eqref{eq:rd}, and \eqref{eq:rs}, that
$
\| \mf r_\phi \| \le   R_{\mrm{se}} \delta(M) + R_i h_{\ds X} + R_d  \rho^2,
$
with $R_{\mrm{se}} = R_s + R_e$. This completes the proof. \hfill $\blacksquare$ 

% \subsection{Density-Aware Farthest Point Sampling (DA-FPS)}

% Consider the data cloud $\mc D = \{ z_i \mid (z_i,z_i^+) \in \mc T \}_{i=1}^N$ and let $\ds X \subset \mc D$. Consider the weighted fill distance 
% \eqnn{
% h_{\ds X, \mc D} = \max_{z\in \mc D} \left( \omega_{\ds X}(z) \min_{\zeta \in \ds X}\| z-\zeta \|_2 \right),
% }
% where $1 \le \omega_{\ds X}(z) \le d$ denotes the local $d$-neighbor density weight. Let $\ds X^\star \in \arg\min_{\ds X \subset \mc D} h_{\ds X, \mc D}$ be an optimal set of $p$ data points and let $\widehat{\ds X} \subset \mc D$ be the set returned by DA-FPS algorithm \cite[Algorithm 1]{climaco2026densityaware}. Then, $h_{\widehat{\ds X}, \mc D} \le 2d\, h_{\ds X^\star, \mc D}$.

\vspace{0.5em}

\noindent \textbf{Proof of Proposition 1:} 
Let $z = [g_1(z)\ \cdots \ g_{nM}(z)]^\top$, where $g_r:\mc Z \to \R$, with $r \in[1:nM]$, denotes the $r$-th coordinate projection. We represent the kernel interpolant of each coordinate function $g_r$ in  $\mrm{span}\{\Psi_1,\cdots,\Psi_p \}$ by
\eqnn{
\hat g_r(z) = {\scriptstyle \sum}_{\ell=1}^p \alpha_{r\ell} \,\Psi_\ell(z) = \alpha_r^\top \bs \Psi(z),
}
where $\alpha_r := [\alpha_{r1} \ \cdots \ \alpha_{rp}]^\top$. The coefficients are chosen to satisfy 
$
\hat g_r(z_m) = g_r(z_m), \quad \forall z_m \in \ds X,\; m \in [1:p].
$
Stacking the coordinate-wise interpolants therefore gives
\eqnn{
\hat z = [\hat g_1(z)\ \cdots \ \hat g_{nM}(z)]^\top = \mf C K_{\ds X}^{-1}\bs \Psi(z),
}
where $\mf C$ is as defined in the statement of the proposition. 
Let $\mc G(\cdot) := \mf C K_{\ds X}^{-1}\bs\Psi(\cdot) $. By the interpolation property, $\mc G(z_\ell)=z_\ell$ for all $z_\ell\in\ds X$, $\ell \in [1:p]$. Hence, for any $\zeta\in\mc Z$, choosing a nearest center $z_\ell\in\ds X$ and using \Cref{lem:rkhs_lip},
\begin{align}
  \! \! \!  \| \zeta - \mc G(\zeta) \| \le \|\zeta -z_\ell \| +\| \mc G(z_\ell)- \mc G(\zeta) \|   \le B_p h_\ds{X},
  \label{eq:Bp}
\end{align}
where $B_p =  1+C_1 \|\mf C K_{\ds X}^{-1} \| \| \bs \Psi\|_{\mc H}$.
Now, 
\begin{align}
z^+ - \hat z^+
&= z^+ - \mf C K_{\ds X}^{-1}\mf A \bs\Psi(z) \notag \\
&={\scriptstyle \bigl(z^+ -\ \mf C K_{\ds X}^{-1}\bs\Psi(z^+)\bigr) + \mf C K_{\ds X}^{-1}\bigl(\bs\Psi(z^+) - \mf A \bs\Psi(z)\bigr) }.
\label{eq:zp_hatz}
\end{align}
Applying \eqref{eq:Bp} at $z^+ \in \mc Z$ and using the residual bound of \Cref{thm1} in \eqref{eq:zp_hatz} gives \eqref{eq:prop3}. \hfill $\blacksquare$ 
% ---- end sections/appendix ----

%% ---- inlined bibliography (was: bibtex on bibl.bib with IEEEtran.bst) ----
% Generated by IEEEtran.bst, version: 1.14 (2015/08/26)

%% ---- end bibliography ----
%\input{sections/10_appendix}
%\input{sections/app_temp}

\end{document}